\documentclass[11pt,oneside]{article}

\usepackage[utf8]{inputenc} \usepackage[T1]{fontenc}    \usepackage{url}            \usepackage{booktabs}       \usepackage{amsfonts}       \usepackage{nicefrac}       \usepackage{microtype}      

\usepackage{lmodern}

\usepackage{amssymb,amsmath,amsthm}
\usepackage{stmaryrd}
\usepackage[short]{optidef}

\usepackage{xcolor}
\usepackage{subfigure,graphicx,epsfig,tikz,caption}
\usepackage[normalem]{ulem}
\usepackage{enumerate}
\usepackage{verbatim}
\usepackage{makeidx,latexsym}
\usepackage[colorlinks=true]{hyperref}

\usepackage{bm}

\usepackage[numbers,sort&compress,square,comma]{natbib}

\usepackage{algorithmic,algorithm}

\usepackage{authblk}
\usepackage[in]{fullpage}
\makeatletter
\def\imod#1{\allowbreak\mkern10mu({\operator@font mod}\,\,#1)}
\makeatother
\usepackage[compact]{titlesec}

\usepackage[english]{babel}
\usepackage{bbm}
\usepackage{pgfplots}

\usepackage{parskip}
    \usepackage{geometry}
\usepackage{thmtools}

    \declaretheorem{theorem}
    \declaretheorem{corollary}
    \declaretheorem{lemma}

    \declaretheoremstyle[qed=$\square$]{definitionwithend}
    \declaretheorem[style=definitionwithend]{definition}
    \declaretheorem[style=definitionwithend]{assumption}
    \declaretheorem[style=definitionwithend]{example}
    \declaretheorem[style=definitionwithend]{remark} \PassOptionsToPackage{numbers, compress}{natbib}

\definecolor{gold}{rgb}{0.85,0.65,0}

\newcommand{\be}{\begin{eqnarray}}
\newcommand{\ee}[1]{\label{#1}\end{eqnarray}}
\newcommand{\ese}{\end{eqnarray*}}
\newcommand{\bse}{\begin{eqnarray*}}
\def\beq{\begin{equation}}
\def\eeq{\end{equation}}

\def\fnote#1{\footnote}

\newcommand{\ceil}[1]{\ensuremath{\left\lceil #1 \right\rceil}}
\newcommand{\abs}[1]{\ensuremath{\left\lvert #1 \right\rvert}}

\newcommand{\by}{\times}
 
\newcommand{\norm}[1]{\ensuremath{\left\lVert #1 \right\rVert}}
\newcommand{\ip}[1]{\ensuremath{\left\langle #1 \right\rangle}}
\newcommand{\grad}{\ensuremath{\nabla}}

\def\R{{\mathbb{R}}}

\def\cR{{\cal R}}
\def\cS{{\cal S}}

\DeclareMathOperator{\Opt}{Opt}

\DeclareMathOperator*{\argmin}{arg\,min}

\DeclareMathOperator{\Diag}{Diag}

\DeclareMathOperator{\tr}{tr}

\DeclareMathOperator{\inter}{int}

\DeclareMathOperator{\conv}{conv}

\DeclareMathOperator{\clconv}{\overline{conv}}

\def\log{\mathop{{\rm log}}}

    \newcommand{\set}[1]{\left\{#1\right\}}

 \newcommand{\Sequal}{\mathfrak{S}}

\begin{document}
\title{The Generalized Trust Region Subproblem: solution complexity and convex hull results}
\author[1]{Alex L.\ Wang}
\author[1]{Fatma K{\i}l{\i}n\c{c}-Karzan}
\affil[1]{Carnegie Mellon University, Pittsburgh, PA, 15213, USA.}
\date{\today}

\maketitle

\begin{abstract}
We consider the Generalized Trust Region Subproblem~(GTRS) of minimizing a nonconvex quadratic objective over a nonconvex quadratic constraint.
A lifting of this problem recasts the GTRS as minimizing a linear objective subject to two nonconvex quadratic constraints.
Our first main contribution is structural: we give an explicit description of the convex hull of this nonconvex set in terms of the generalized eigenvalues of an associated matrix pencil.
This result may be of interest in building relaxations for nonconvex quadratic programs.
Moreover, this result allows us to reformulate the GTRS as the minimization of two convex quadratic functions in the original space.
Our next set of contributions is algorithmic: we present an algorithm for solving the GTRS up to an $\epsilon$ additive error based on this reformulation.
We carefully handle numerical issues that arise from inexact generalized eigenvalue and eigenvector computations and establish explicit running time guarantees for these algorithms. Notably, our algorithms run in \emph{linear (in the size of the input) time}. Furthermore, our algorithm for computing an $\epsilon$-optimal solution has a slightly-improved running time dependence on $\epsilon$ over the state-of-the-art algorithm. Our analysis shows that the dominant cost in solving the GTRS lies in solving a generalized eigenvalue problem---establishing a natural connection between these problems. Finally, generalizations of our convex hull results allow us to apply our algorithms and their theoretical guarantees directly to equality-, interval-, and hollow-constrained variants of the GTRS. This gives the first linear-time algorithm in the literature for these variants of the GTRS.
 \end{abstract}

\section{Introduction}\label{sec:Intro}
In this paper, we study the \emph{Generalized Trust-Region Subproblem}~(GTRS),
which is defined as 
\begin{align}\label{eq:GTRS}
\Opt\coloneqq \inf_{x\in\R^n}\set{q_0(x) :~ q_1(x) \leq 0},
\end{align}
where $q_0:\R^n\to\R$ and $q_1:\R^n\to\R$ are general quadratic functions of the form $q_i(x) = x^\top A_i x + 2b_i^\top x + c_i$. Here, $A_i\in \R^{n\by n}$ are symmetric matrices, $b_i\in\R^n$ and $c_i\in\R$. We are interested, in particular, in the case where $q_0$ and $q_1$ are both nonconvex, i.e., $A_i$ has at least one negative eigenvalue for both $i=0,1$.

Problem~\eqref{eq:GTRS}, introduced and studied by \citet{more1993generalizations,SternWolkowicz1995}, generalizes the classical \emph{Trust-Region Subproblem}~(TRS)~\cite{Conn.et.al.2000} in which one is asked to optimize a nonconvex quadratic objective over a Euclidean ball. The TRS is an essential ingredient of trust-region methods that are commonly used to solve continuous nonconvex optimization problems~\cite{Conn.et.al.2000,NocedalWright2000numerical,PongWolkowicz2014} and also arises in applications such as robust optimization~\cite{BenTalDenHertog2014,Ho-NguyenKK2016RO}.
On the other hand, the GTRS has applications in nonconvex quadratic integer programs, signal processing, and compressed sensing; see \cite{BuchheimDSPP2013,adachi2019eigenvalue,jiang2020linear} and references therein for more applications.

Although the TRS, as stated, is nonlinear and nonconvex, it is well-known that its semidefinite programming~(SDP) relaxation is exact. Consequently, the TRS and a number of its variants can be solved in polynomial time via SDP-based techniques~\cite{Rendl_Wolkowicz_97,FortinWolkowicz2004} or using specialized nonlinear algorithms~\cite{Gould_LRT_99,More_Sorensen_83}. In fact, 
custom iterative methods with linear (in the size of the input) running times have been shown in a few works. \citet{HazanKoren2016} proposed an algorithm to solve the TRS (as well as the GTRS when $A_1$ is positive definite) based on repeated approximate eigenvector computations.  This algorithm runs in time
\begin{align}
\label{eq:hazan_koren_runtime}
\tilde O\left(\frac{N\sqrt{\kappa_{\text{HK}}}}{\sqrt{\epsilon}} \log\left(\frac{n}{p}\right)\log\left(\frac{\kappa_{\text{HK}}}{\epsilon}\right)\right),
\end{align}
where $N$ is the number of nonzero entries in the matrices $A_0$ and $A_1$, $\epsilon$ is the additive error, $n$ is the dimension of the problem, $p$ is the failure probability, and $\kappa_{\text{HK}}$ is a condition number.
This was the first algorithm in the literature 
shown to achieve a linear time complexity.
Here, and in the remainder of the paper, the term ``linear'' is used to describe running times that scale at most linearly with $N$ but may depend arbitrarily on its other parameters.
Afterwards, \citet{Ho-NguyenKK2017} presented another linear-time algorithm for the TRS with a slightly better overall complexity, eliminating the $\log(\kappa_{\text{HK}}/{\epsilon})$ term. Their approach reformulates the TRS as minimizing a \textit{convex} quadratic objective over the Euclidean ball, and solving the resulting smooth convex optimization problem via Nesterov's accelerated gradient descent method. In contrast to  \cite{HazanKoren2016}, this convex reformulation approach requires only a single minimum eigenvalue computation. \citet{WangXia2016} also suggested using Nesterov's algorithm in the case of the interval-constrained TRS.  

The GTRS shares a number of nice properties of the TRS. For example, by the S-lemma, it is well-known that the GTRS also admits an exact SDP reformulation under the Slater condition~\cite{FradkovYakubovich1979,PolikTerlaky2007}.
Thus, while quadratically-constrained quadratic programming is NP-hard in general, there are polynomial-time SDP-based algorithms for solving the GTRS.
Nevertheless, the relatively large computational complexity of SDP-based algorithms prevents them from being applied as a black box to solve large-scale instances of the GTRS. A variety of custom approaches have been developed to solve the GTRS; for earlier work on this domain see \cite{more1993generalizations,SternWolkowicz1995,feng2012duality} and references therein.

One line of work has developed algorithms for solving the GTRS when the matrices $A_0$ and $A_1$ are simultaneously diagonalizable~(SD) (see \citet{jiang2016simultaneous} and references therein for background on the SD condition). 
Under the SD condition, along with certain restrictions on the quadratics $q_0$ and $q_1$, \citet{BenTalTeboulle1996} provide a reformulation of the interval-constrained GTRS as a convex minimization problem with linear constraints. 
More recently, \citet{BenTalDenHertog2014} show that there is a second order cone programming~(SOCP) reformulation of the GTRS in a lifted space under the SD condition.
Subsequent work by \citet{locatelli2015some} extends \citet{BenTalDenHertog2014} by illustrating some additional settings in which the SOCP reformulation is tight. 
Under the SD condition, \citet{fallahi2018minimizing} exploit the separable structure of the problem and, using Lagrangian duality, they suggest a solution procedure based on solving a univariate convex minimization problem. 
\citet{salahi2018efficient} derive an algorithm for solving the interval-constrained GTRS by exploiting the structure of the dual problem under the SD condition.
By applying a simultaneous block diagonalization approach, \citet{jiang2018socp} generalize \citet{BenTalDenHertog2014} and provide an SOCP reformulation for the GTRS in a lifted space when the problem has a finite optimal value. Their methods apply even when $q_0$ and $q_1$ do \textit{not} satisfy the SD condition. They further derive a closed-form solution when the SD condition fails and examine the case of interval-\ or equality-constrained GTRS.
In this line of work, it is often assumed implicitly that $A_0$ and $A_1$ are already diagonal or that a simultaneously-diagonalizing basis can be computed.
The only method that we know of for computing such a basis relies on exact matrix eigen-decomposition.
Thus, although experiments have been presented~\cite{salahi2018efficient,jiang2018socp} suggesting that such algorithms (where exact procedures are replaced by numerical ones) may perform well, theoretical guarantees have yet to be established.
Furthermore, the large cost of matrix eigen-decomposition prevents the application of these algorithms to large-scale instances of the GTRS. 

A second line of work has explored the connections between the GTRS and generalized eigenvalues of the matrix pencil $A_0 + \gamma A_1$. 
These works all assume a \emph{regularity condition} about the matrix pencil: there exists a $\gamma\geq 0$ such that $A_0 +\gamma A_1$ is either positive definite or positive semidefinite.\footnote{In fact, this assumption can be made without loss of generality; see Remark~\ref{rem:no_psd_trivial}.}
\citet{PongWolkowicz2014} study the optimality structure of the GTRS and propose a generalized-eigenvalue-based algorithm which exploits this structure. Unfortunately, an explicit running time is not presented in \cite{PongWolkowicz2014}.
\citet{adachi2019eigenvalue} present another generalized-eigenvalue-based algorithm motivated by similar observations. The dominant costs present in this algorithm come from computing a pair of generalized eigenvalues and solving a linear system. Ignoring issues of exact computations, the runtime of this algorithm is $O(n^3)$.
\citet{jiang2019novel} show how to reformulate the GTRS as a convex quadratic program in terms of generalized eigenvalues. They establish that a saddle-point-based first-order algorithm can be used to solve the reformulation within an $\epsilon$ additive error in $O(1/\epsilon)$ time.
In this line of work, it is often assumed that the generalized eigenvalues are given or can be computed exactly. In particular, theoretical guarantees have not yet been given regarding how these algorithms perform when only approximate generalized eigenvalue computations are available.
This is of interest as, in practice, we cannot hope to numerically compute generalized eigenvalues exactly; see also the discussion in Section 4.1 in \cite{jiang2020linear}. We would like to remark that numerical experiments in these papers~\cite{jiang2019novel,adachi2019eigenvalue,PongWolkowicz2014} have suggested that algorithms motivated by these ideas may perform well even using only approximate generalized eigenvalue computations.

The very recent work of \citet{jiang2020linear} presents an algorithm for solving the GTRS up to an $\epsilon$ additive error in the objective with high probability under the regularity condition. This algorithm relies on machinery developed by \cite{HazanKoren2016} for solving the TRS and differs from previous algorithms in that it does not assume the ability to compute a simultaneously-diagonalizing basis or generalized eigenvalues. The running time of this algorithm is
\begin{align}
\label{eq:JL_runtime}
\tilde O\left(\frac{N\phi^3}{\sqrt{\epsilon\,\xi_{\text{JL}}^5}}\log\left(\frac{n}{p}\right)\log\left(\frac{\phi}{\epsilon\,\xi_{\text{JL}}}\right)^2\right),
\end{align}
where $N$ is the number of nonzero entries in $A_0$ and $A_1$, $\epsilon$ is the additive error, $n$ is the dimension, $p$ is the failure probability, and $(\phi, \xi_{\text{JL}})$ are a pair of parameters measuring the regularity of the GTRS.
In particular, this algorithm is able to take advantage of sparsity in the description of the quadratic functions.
To our knowledge, this is the first provably linear-time algorithm for the GTRS to be presented in the literature.

In this paper, we derive a new algorithm for the GTRS based on a convex quadratic reformulation in the original space. This algorithm can also be applied to variants of the GTRS with interval, equality, or hollow constraints.
The basic idea in our approach relies on the fact that we can provide exact (closed) convex hull characterizations of the epigraph of the GTRS.
We summarize our results below and provide an outline of the paper.

\begin{enumerate}[(i)]
	\item We rewrite the GTRS with a linear objective
	\begin{align}\label{eq:GTRSepi}
	\Opt = \inf_{(x,t)}\set{t:\, (x,t)\in \cS},
	\end{align}
	where the set $\cS$ is defined as
	\begin{align}\label{eq:cSdef}
	\cS\coloneqq \set{(x,t)\in\R^{n+1}:\, \begin{array}
		{l}
		q_0(x)\leq t\\
		q_1(x)\leq 0
	\end{array}}.
	\end{align}
	As the objective in \eqref{eq:GTRSepi} is linear, we can take either the convex hull or closed convex hull of the feasible domain. Then,
	\begin{align*}
	\Opt = \inf_{x,t}\set{t :~ (x,t)\in \conv(\cS)}=\inf_{x,t}\set{t :~ (x,t)\in \overline{\conv}(\cS)}.
	\end{align*}
	In Section~\ref{sec:ConvHull}, we give an explicit description of the set $\conv(\cS)$ (respectively, $\clconv(\cS)$).
	Specifically, we show that when the respective assumptions are satisfied, $\conv(\cS)$ and $\clconv(\cS)$ can both be described in terms of two convex quadratic functions determined by the generalized eigenvalue structure of the matrix pencil $A_0 + \gamma A_1$. 
	We note that these convex hull results may be of independent interest in building relaxations and/or algorithms for nonconvex quadratic programs with or without integer variables.
	As an immediate consequence of these (closed) convex hull results, we can reformulate the GTRS as the minimization of the maximum of two convex quadratics.
	This convex reformulation was previously discovered by \citet{jiang2019novel} by considering the Lagrangian dual and proving a zero duality gap. Our approach shows that the reformulation is tight for a very intuitive reason --- the convex hull of the epigraph is exactly characterized by the convex quadratics used in the reformulation.
	\item
	The proofs in Section~\ref{sec:ConvHull} actually imply stronger convex hull results: under the same assumptions, the (closed) convex hull of $\cS$ is generated by points in $\cS$ where the constraint $q_1(x)\leq 0$ is tight.
	This observation immediately leads to interesting consequences, which we detail in Section~\ref{sec:nonintersecting_constraints}.
	Specifically, we extend our (closed) convex hull results to handle epigraph sets that arise when additional nonintersecting constraints are imposed on the GTRS. This will allow us to extend our algorithms to variants of the GTRS present in the literature~\cite{BenTalTeboulle1996,BenTalDenHertog2014,jiang2018socp,jiang2019novel,PongWolkowicz2014,salahi2018efficient,	Ho-NguyenKK2017,
	yang2018quadratic,SternWolkowicz1995,more1993generalizations}.
	Specifically, this generalization allows us to handle interval-, equality-, and hollow-constrained GTRS.

	\item In Section~\ref{sec:alg}, we give a careful analysis of the numerical issues that come up for an algorithm based on the above ideas. At a high level, we show that by approximating the generalized eigenvalues sufficiently well, the perturbed convex reformulation is within a small additive error of the true convex reformulation.
	Then, by leveraging the concavity of the function $\lambda_{\min}(A_0 + \gamma A_1)$, in the variable $\gamma$, we show how to approximate the necessary generalized eigenvalues efficiently.
	We believe this subroutine and the theoretical guarantees we present for it may also be of independent interest in other contexts. 
	Next, we utilize an algorithm proposed by \citet[Section 2.3.3]{nesterov2018lectures} for solving general minimax problems with smooth components to solve our convex reformulation with a convergence rate of $\tilde O(1/\sqrt{\epsilon})$.
	This contrasts the approach taken by \citet{jiang2019novel} that analyzes a saddle-point-based first-order algorithm and results in a convergence rate of $O(1/\epsilon)$.
	In order to apply the algorithm proposed by Nesterov, we establish that the gradient mapping step can be computed efficiently in our context.
	Finally, relying on our convex hull characterization, we show how to recover an approximate solution of the GTRS using only approximate eigenvectors.

	We present two algorithms (Algorithms~\ref{alg:GTRS_outer}~and~\ref{alg:rounding}). The former finds an $\epsilon$-optimal value and the latter finds an $\epsilon$-optimal feasible solution. In other words, the former returns a scalar in $[\Opt, \Opt+\epsilon]$ and the latter returns a vector $x$ in the feasible region with $q_0(x)\in[\Opt,\Opt+\epsilon]$. Their running times are
	\begin{align}
	\label{eq:us_runtime}
	\tilde O\left(\frac{N\kappa^{3/2}\sqrt{\zeta}}{\sqrt{\epsilon}}\log\left(\frac{n}{p}\right)\log\left(\frac{\kappa}{\epsilon}\right)\right)
	,\quad
	\tilde O\left(\frac{N\kappa^{2}\sqrt{\zeta}}{\sqrt{\epsilon}}\log\left(\frac{n}{p}\right)\log\left(\frac{\kappa}{\epsilon}\right)\right),
	\end{align}
	respectively. Here, $\xi$, $\zeta$, and $\kappa$ are regularity parameters of the matrix pencil $A_0 + \gamma A_1$ (see Definition~\ref{def:regularity}).
	Comparing \eqref{eq:us_runtime} and \eqref{eq:hazan_koren_runtime}, we see that our running times match the dependences on $N$, $n$, $\epsilon$, and $p$ from the algorithm for the TRS presented by \citet{HazanKoren2016}.
	Comparing \eqref{eq:us_runtime} (specifically the running time for finding an $\epsilon$-optimal solution) and \eqref{eq:JL_runtime}, we see that our running time matches the linear dependence on $N$ and improves the dependence on $\epsilon$ by a logarithmic factor from the running time presented by \citet{jiang2020linear}.
	The dependences on the regularity parameters in the two running times are incomparable (see Remark~\ref{rem:regularity_comparison}) but there exist examples where our running time gives a polynomial-order improvement upon the running time presented by \citet{jiang2020linear} (see Remark~\ref{rem:runtime_comparison}).
	
	In comparison to the approach taken by \citet{jiang2020linear}, we believe our approach is conceptually simpler and more straightforward to implement. In particular our approach directly solves the GTRS in the primal space as opposed to solving a feasibility version of the dual problem. Moreover, our analysis highlights the connection
between the GTRS and generalized eigenvalue problems, and in fact demonstrates that the dominant cost in solving the GTRS is the cost of solving a generalized eigenvalue problem.

	In our running times \eqref{eq:us_runtime}, the large dependence on the regularity parameters arises from the error that is introduced as a result of inexact generalized eigenvalue and eigenvector computations.
	We illustrate that our algorithms can be substantially sped up if we have access to exact generalized eigenvalue and eigenvector methods. In particular, we show that when $A_0$ and $A_1$ are diagonal, we can compute an $\epsilon$-optimal solution to the GTRS in time
	\begin{align*}
	O\left(\frac{N\kappa\sqrt{\zeta}}{\sqrt{\epsilon}}\right).
	\end{align*}
	
	As mentioned previously, the generalizations of our convex hull results allow us to apply our algorithms to variants of the GTRS. In particular, our algorithms can be applied without change to interval-, equality-, or hollow-constrained GTRS.
\end{enumerate}

Our study of the convex hull of the epigraph of GTRS is inspired by convex hull results in related contexts.
The recent work of \citet{Ho-NguyenKK2017} gives a characterization on the convex hull of the epigraph of the TRS. In particular, under the assumption that $A_1$ is positive definite, \citet[Theorem 3.5]{Ho-NguyenKK2017} give the explicit closed convex hull characterization of the set $\cS$. 
In this respect, one can view our developments on the (closed) convex hull of $\cS$ when neither $A_0$ nor $A_1$ is positive semidefinite as complementary to the results of \citet[Section 3]{Ho-NguyenKK2017}. 
Notably, in contrast to \cite[Section 3]{Ho-NguyenKK2017}, we have to handle a number of issues that arise due to the recessive directions of the nonconvex domain.
The papers by \citet{yildiran2009convex,modaresi2017convex} are also closely related to our convex hull results. \citet{yildiran2009convex} studies the convex hull of the intersection of two \emph{strict} quadratic inequalities (note that the resulting set is open) under the milder regularity condition that there exists $\gamma\geq0$ such that $A_0+\gamma A_1$ is positive semidefinite, and \citet{modaresi2017convex} analyze conditions under which one can safely take the closure of the sets in \citet{yildiran2009convex} and still obtain the desired closed convex hull results. 
In contrast, our analysis leverages the additional structure present in an epigraph set to give a more direct proof of the convex hull result.
Furthermore, as our analysis is constructive (given $x\in \conv(\cS)$, we show how to find two points $x_1$, $x_2\in\cS$ such that $x\in[x_1,x_2]$), it immediately suggests a rounding procedure (given a solution to the convex reformulation, we show how to find a solution to the original GTRS).
This contrasts the analysis in \citet{yildiran2009convex}, where such a rounding procedure is not obvious.
Moreover, our analysis provides a more refined result that easily extends to variants of the GTRS with non-intersecting constraints. 
Finally, we would like to mention related work on convex hulls of sets defined by second-order cones~(SOCs). \citet{BKK14} study the convex hull of the intersection of a convex and nonconvex quadratic or the intersection of an SOC with a nonconvex quadratic. Similarly, the convex hull of the a two-term disjunction applied to an SOC or its cross section has received much attention (see \cite{BKK14,KKY15} and references therein). As our focus has been on the case where neither $A_0$ nor $A_1$ is positive semidefinite, we view our developments as complementary to these results. 

\emph{Notation}. 
Given a symmetric matrix, $A\in\R^{n\by n}$, let $\lambda_{\min}(A)$ and $\lambda_{\max}(A)$ denote the minimum and maximum eigenvalues of $A$.
We write $A\succeq 0$ (respectively, $A\succ 0$) if $A$ is positive semidefinite (respectively, positive definite).
Let $\norm{A}$ denote the spectral norm of $A$, i.e. $\norm{A}=\max\set{\abs{\lambda_{\min}(A)},\abs{\lambda_{\max}(A)}}$.
Let $\det(A)$ and $\tr(A)$ denote the determinant and trace of $A$.
For $a\in\R^n$, let $\Diag(a)$ denote the diagonal matrix $A\in\R^{n\by n}$ with diagonal entries $A_{i,i} = a_i$.
Let $I_n$ be the $n\times n$ identity matrix.
For a set $S \subseteq \R^n$, let $\conv(S)$ and $\clconv(S)$ be the convex hull and closed convex hull of $S$, respectively.
For $x\in\R^n$, let $\norm{x}$ be its Euclidean norm.
For $x\in\R^n$ and $r\geq 0$, let $B(x,r)$ be the closed ball of radius $r$ centered at $x$, i.e., $B(x,r) = \set{y\in\R^n :\, \norm{x-y}\leq r}$.
Let $\R_+$ denote the nonnegative reals. Let $\grad$ denote the gradient operator. We will use $\tilde O$-notation to hide $\log\log$-factors in running times. 

\section{Convex hull characterization}\label{sec:ConvHull}
In this section we discuss our (closed) convex hull results.
We will aggregate the objective function $q_0$ with the constraint $q_1$ using a nonnegative aggregation weight to derive relaxations of the set $\cS$. We then show that under a mild assumption the (closed) convex hull of $\cS$ can be described by two convex quadratic functions obtained from this aggregation technique.

Let $q:\R\times\R^n\to\R$ be defined as
\begin{align*}
q(\gamma,x) \coloneqq q_0(x) + \gamma q_1(x).
\end{align*}
Let $A:\R\to\R^{n\by n}$ be defined as $A(\gamma) \coloneqq A_0 + \gamma A_1$. Similarly define $b(\gamma)$ and $c(\gamma)$.
In particular, $q(\gamma,x) = x^\top A(\gamma) x + 2 b(\gamma)^\top x + c(\gamma)$.
We stress that while $q(0,x) = q_0(x)$, we have $q(1,x) = q_0(x)+q_1(x)$ which is not equal to $q_1(x)$ in general.

Note that $q(\gamma,x)$ is linear in its first argument and quadratic in its second argument. This structure plays a large role in our analysis.

In order to derive valid relaxations to $\cS$ based on aggregation, we will consider only nonnegative $\gamma$ in the remainder of the paper. For $\gamma\geq 0$, define
\begin{align*}
\cS(\gamma) \coloneqq \set{(x,t)\in\R^{n+1}:~ q(\gamma,x)\leq t}.
\end{align*}
Note that $\cS \subseteq \cS(\gamma)$ holds for all $\gamma\geq 0$. 
Furthermore, it is clear that $q_0(x) \leq t$ and $q_1(x)\leq 0$ if and only if $q(\gamma,x)\leq t$ for all $\gamma\geq 0$.
Thus, we can rewrite $\cS$ as
\begin{align*}
\cS \coloneqq \set{(x,t)\in\R^{n+1}:~ \begin{array}{l}
	q_0(x)\leq t\\
	q_1(x) \leq 0
\end{array}}
=\bigcap_{\gamma\geq0} \cS(\gamma).
\end{align*}

Note that the set $\cS(\gamma)$ is convex if and only if $A(\gamma)\succeq 0$. We will define $\Gamma$ to be these $\gamma$ values, i.e.,
\begin{align*}
\Gamma\coloneqq \set{\gamma\in\R_{+}:~ A(\gamma)\succeq 0}.
\end{align*}

Note that $\Gamma$ is a closed (possibly empty) interval. When this interval is nonempty, we will write it as $\Gamma = [\gamma_-,\gamma_+]$.

We use the following two assumptions in our convex hull characterizations:
\begin{assumption}\label{as:pd}
The matrices $A_0$ and $A_1$ both have negative eigenvalues and there exists a $\gamma^*\geq 0$ such that $A(\gamma^*) \succ 0$.
\end{assumption}

\begin{assumption}\label{as:psd}
The matrices $A_0$ and $A_1$ both have negative eigenvalues and there exists a $\gamma^*\geq 0$ such that $A(\gamma^*) \succeq 0$.
\end{assumption}

\begin{remark}
	\label{rem:no_psd_trivial}
	We claim that the case where $A_0$ and $A_1$ both have negative eigenvalues but do not satisfy either of the above assumptions is not interesting. In particular if $A_0$ and $A_1$ both have negative eigenvalues and $A(\gamma)\not\succeq 0$ for all $\gamma\geq 0$, then it is easy to show (apply the S-lemma then note that $A_0$ has a negative eigenvalue) that $\conv(\cS)=\R^{n+1}$. Consequently, the optimal value of the GTRS is always $-\infty$ in this case.
	
	The assumption that there exists a $\gamma^*\geq 0$ such that $A(\gamma^*)\succeq 0$ is  made in most of the present literature on the GTRS~\cite{BenTalDenHertog2014,adachi2019eigenvalue,jiang2016simultaneous,jiang2019novel,jiang2020linear,PongWolkowicz2014,salahi2018efficient} and convex hulls of the intersection of two quadratics~\cite{yildiran2009convex,modaresi2017convex} either implicitly (for example, by assuming that an optimizer exists or that the optimal value is finite) or explicitly.

	It is well-known that Assumption~\ref{as:pd} implies that $A_0$ and $A_1$ are simultaneously diagonalizable. Even so, we will refrain from assuming that our matrices are diagonal and opt to work on a general basis. We choose to do this as the proofs of our convex hull results will serve as the basis for our algorithms, which do not have access to a simultaneously-diagonalizing basis.
	\end{remark}

\begin{remark}\label{rem:assumptions}
Assumptions~\ref{as:pd} and \ref{as:psd} each imply that $\Gamma$ is nonempty and, consequently, that $\gamma_-$ and $\gamma_+$ exist. In addition, as $A(\gamma_-)$ and $A(\gamma_+)$ are both on the boundary of the positive semidefinite cone, they both have zero as an eigenvalue.

Under Assumption~\ref{as:pd}, the existence of some $\gamma^*\geq 0$ such that $A(\gamma^*) \succ 0$ implies that $\gamma_-<\gamma^* <\gamma_+$ and hence $\gamma_-$ and $\gamma_+$ are distinct.
Furthermore, as $\gamma^*\in(\gamma_-,\gamma_+)$, we have  $d^\top A (\gamma_-)d = d^\top A(\gamma_+) d = 0$ if and only if $d = 0$.

In contrast, under Assumption~\ref{as:psd}, it is possible to have $\gamma_- =\gamma^*=\gamma_+$ and $\Gamma=\{\gamma^*\}$.
\end{remark}

Finally, define $\Sequal$ to be the subset of $\cS$ where the constraint $q_1(x)\leq 0$ is tight.
\begin{align*}
\Sequal \coloneqq \set{(x,t)\in\R^{n+1}:~ \begin{array}{l}
	q_0(x) \leq t\\
	q_1(x) = 0
\end{array}}.
\end{align*}
When either Assumption~\ref{as:pd}~or~\ref{as:psd} holds, $A_1$ has both positive and negative eigenvalues so that $\Sequal$ is nonempty.

We now state our (closed) convex hull results:
\begin{theorem}\label{thm:pd_implies_conv_hull}
Under Assumption~\ref{as:pd}, we have
\begin{align*}
\conv(\cS) =\conv(\Sequal) = \cS(\gamma_-) \cap \cS(\gamma_+).
\end{align*}
In particular,
\begin{align*}
\min_{x\in\R^n}\set{q_0(x):~ q_1(x) \leq 0}
= \min_{x\in\R^n} \max\set{q(\gamma_-,x),q(\gamma_+,x)}.
\end{align*}
\end{theorem}

\begin{theorem}\label{thm:psd_implies_conv_hull}
Under Assumption~\ref{as:psd}, we have
\begin{align*}
{\clconv}(\cS) ={\clconv}(\Sequal) = \cS(\gamma_-) \cap \cS(\gamma_+).
\end{align*}
In particular,
\begin{align*}
\inf_{x\in\R^n}\set{q_0(x):~ q_1(x) \leq 0}
= \inf_{x\in\R^n} \max\set{q(\gamma_-,x),q(\gamma_+,x)}.
\end{align*}
\end{theorem}

\begin{remark}
The convex reformulation given in the second part of Theorem~\ref{thm:psd_implies_conv_hull} was first proved by \citet{jiang2019novel} using a different argument without relying on the convex hull structure of the underlying sets. In contrast, the first part of Theorem~\ref{thm:psd_implies_conv_hull} establishes a fundamental convex hull result highlighting the crux of why such a convex reformulation is possible.
\end{remark}

We present the proof of Theorem~\ref{thm:pd_implies_conv_hull} in Section~\ref{sec:proof_pd}. The proof of Theorem~\ref{thm:psd_implies_conv_hull} is presented in Section~\ref{sec:proof_psd} and relies on Theorem~\ref{thm:pd_implies_conv_hull}.

\subsection{Proof of Theorem~\ref{thm:pd_implies_conv_hull}}
\label{sec:proof_pd}

\begin{lemma}
\label{lem:psd_implies_convex_closed}
The set $\cS(\gamma)$ is convex and closed for all $\gamma\in\Gamma$.
\end{lemma}
\begin{proof}
Let $\gamma\in\Gamma$ and recall the definition of $\cS(\gamma)$.
\begin{align*}
\cS(\gamma) &\coloneqq \set{(x,t)\in\R^{n+1} :~ q(\gamma,x)\leq t}\\
&=\set{(x,t)\in\R^{n+1} :~ x^\top A(\gamma)x + 2b(\gamma)^\top x + c(\gamma) \leq t}
\end{align*}
By the definition of $\Gamma$, we have $A(\gamma)\succeq 0$. Thus, the constraint defining $\cS(\gamma)$ is convex in $(x,t)$, and we conclude that $\cS(\gamma)$ is convex. 
Closedness of $\cS(\gamma)$ follows by noting that it is the preimage of $(-\infty,0]$ under a continuous map.
\end{proof}

\begin{lemma}\label{lem:pd:Ssubset}
Suppose $\Gamma$ is nonempty and write $\Gamma=[\gamma_-,\gamma_+]$. Then, $\conv(\cS) \subseteq \cS(\gamma_-) \cap \cS(\gamma_+)$.
\end{lemma}
\begin{proof}
Note that $\cS = \bigcap_{\gamma\geq 0} \cS(\gamma) \subseteq \cS(\gamma_-)\cap \cS(\gamma_+)$. The result then follows by taking the convex hull of each side and noting that both $\cS(\gamma_-)$ and $\cS(\gamma_+)$ are convex by Lemma~\ref{lem:psd_implies_convex_closed}.
\end{proof}

The bulk of the work in proving Theorem~\ref{thm:pd_implies_conv_hull} lies in the following result.
\begin{lemma}\label{lem:pd:Ssupset}
Under Assumption~\ref{as:pd},  we have $\cS(\gamma_-)\cap \cS(\gamma_+)\subseteq\conv(\Sequal)$.
\end{lemma}
\begin{proof}
Let $(\hat x,\hat t)\in \cS(\gamma_-)\cap \cS(\gamma_+)$. We will show that $(\hat x,\hat t)\in\conv(\Sequal)$. We split the analysis into three cases: (i) $q_1(\hat x)=0$,  (ii) $q_1(\hat x)>0$, and (iii) $q_1(\hat x)<0$. 

\begin{enumerate}[(i)]
\item If $q_1(\hat x) = 0$, then $q_0(\hat x) = q_0(\hat x) + \gamma_- q_1(\hat x)= q(\gamma_-,\hat x)$. As $(\hat x,\hat t)\in \cS(\gamma_-)$ by assumption, we deduce that $q(\gamma_-,\hat x)\leq \hat t$. Combining these inequalities, we have that $q_0(\hat x)= q(\gamma_-, \hat x) \leq \hat t$ and that $(\hat x,\hat t)\in \Sequal$.

\item Now suppose $q_1(\hat x)>0$.
Let $d\neq 0$ such that $d^\top A(\gamma_+) d = 0$ (such a vector $d$ exists as $A(\gamma_+)$ has zero as an eigenvalue; see Remark~\ref{rem:assumptions}) and define $e \coloneqq 2\left(\hat x^\top A(\gamma_+)d + b(\gamma_+)^\top d\right)$.
We modify $(\hat x,\hat t)$ along the direction $(d,e)$: For $\alpha\in\R$, let $(\hat x_\alpha,\hat t_\alpha) \coloneqq  (\hat x+\alpha d, \hat t+\alpha e)$.
We will show that there exist $\alpha_1<0<\alpha_2$ such that $(\hat x_{\alpha_i},\hat t_{\alpha_i})\in \Sequal$ for $i=1,2$, whence $(\hat x,\hat t)\in\conv(\Sequal)$.

We study the behavior of the expressions $q(\gamma_-,\hat x_\alpha)-\hat t_\alpha$ and $q(\gamma_+,\hat x_\alpha)-\hat t_\alpha$ as functions of $\alpha$.
A short calculation shows that for any $\alpha\in\R$, we have
\begin{align}\label{eq:thm:pd:q_lambda_+}
&q(\gamma_+, \hat x_\alpha) - \hat t_\alpha\notag\\
&\qquad= \left(q(\gamma_+,\hat x) - \hat t\right) + 2\alpha\left(\hat x^\top A(\gamma_+)d + b(\gamma_+)^\top d - e/2\right) + \alpha^2d^\top A(\gamma_+)d \notag \\
&\qquad= q(\gamma_+,\hat x) - \hat t,
\end{align}
where the last equation follows from the definition of $e$.
Thus, $q(\gamma_+,\hat x_\alpha)-\hat t_\alpha$ is constant in $\alpha$.
Next, we compute
\begin{align*}
&q(\gamma_-, \hat x_\alpha) - \hat t_\alpha\\
&\qquad= \left(q(\gamma_-,\hat x) - \hat t\right) + 2\alpha\left(\hat x^\top A(\gamma_-)d + b(\gamma_-)^\top d - e/2\right) + \alpha^2d^\top A(\gamma_-)d.
\end{align*}
As $d\neq 0$ and $d^\top A(\gamma_+)d=0$, we deduce that $d^\top A(\gamma_-)d\neq 0$ (see Remark~\ref{rem:assumptions}). Then, as  $A(\gamma_-)\succeq 0$, we have that $d^\top A(\gamma_-)d>0$. Hence, $q(\gamma_-, \hat x_\alpha) - \hat t_\alpha$ is strongly convex in $\alpha$. 

Note that 
\[
q(\gamma_-,\hat x) = q_0(\hat x) + \gamma_- q_1(\hat x) < q_0(\hat x) + \gamma_+ q_1(\hat x) = q(\gamma_+,\hat x),
\]
where the inequality follows from the fact that $\gamma_-<\gamma_+$ and $q_1(\hat x)>0$. 
Therefore, $q(\gamma_-,\hat x) - \hat t < q(\gamma_+,\hat x) - \hat t$. 
Thus, there are values $\alpha_1<0<\alpha_2$ such that $q(\gamma_-,\hat  x_{\alpha_i}) - \hat t_{\alpha_i} =q(\gamma_+, \hat x_{\alpha_i}) - \hat t_{\alpha_i}$ for $i=1,2$.

It remains to show that $(\hat x_{\alpha_i},\hat t_{\alpha_i})\in\Sequal$ for $i=1,2$.
This follows immediately because for $i=1,2$, we have
\begin{align*}
q_1(\hat x_{\alpha_i}) &= \frac{1}{\gamma_+-\gamma_-}\left(q(\gamma_+,\hat x_{\alpha_i}) - q(\gamma_-,\hat x_{\alpha_i})\right) = 0.
\end{align*}
Then, applying \eqref{eq:thm:pd:q_lambda_+} and recalling that $q(\gamma_+,\hat x)\leq \hat t$, we have
\begin{align*}
q_0(\hat x_{\alpha_i}) &= q(\gamma_+,\hat x_{\alpha_i}) - \gamma_+ q_1(\hat x_{\alpha_i}) =q(\gamma_+,\hat x_{\alpha_i})\leq \hat t_{\alpha_i}.
\end{align*}

\item The final case is symmetric to case (ii), thus we will only sketch its proof.

Suppose $q_1(\hat x)<0$.
Let $d\neq 0$ such that $d^\top A(\gamma_-) d = 0$ and define $e \coloneqq 2\left(\hat x^\top A(\gamma_-)d + b(\gamma_-)^\top d\right)$.
For $\alpha\in\R$, let $(\hat x_\alpha,\hat t_\alpha) \coloneqq (\hat x+\alpha d, \hat t+\alpha e)$.

A short calculation shows  that for any $\alpha\in\R$, we have
\begin{align*}
&q(\gamma_-, \hat x_\alpha) - \hat t_\alpha\\
&\qquad= \left(q(\gamma_-,\hat x) - \hat t\right) + 2\alpha\left(\hat x^\top A(\gamma_-)d + b(\gamma_-)^\top d - e/2\right) + \alpha^2d^\top A(\gamma_-)d\\
&\qquad= q(\gamma_-,\hat x) - \hat t.
\end{align*}
Similarly, for any $\alpha\in\R$, 
\begin{align*}
&q(\gamma_+, \hat x_\alpha) - \hat t_\alpha\\
&\qquad= \left(q(\gamma_+,\hat x) - \hat t\right) + 2\alpha\left(\hat x^\top A(\gamma_+)d + b(\gamma_+)^\top d - e/2\right) + \alpha^2d^\top A(\gamma_+)d.
\end{align*}
As $d^\top A(\gamma_-)d=0$ and $d\neq 0$, Assumption~\ref{as:pd} implies that $d^\top A(\gamma_+)d>0$. We see that $q(\gamma_+,\hat x_\alpha)-\hat t_\alpha$ is strongly convex in $\alpha$. 
As $q_1(\hat x)<0$, we have $q(\gamma_+,\hat x) -\hat  t < q(\gamma_-,\hat x) - \hat t$. 
Thus, there are values $\alpha_1<0<\alpha_2$ such that $q(\gamma_+, \hat x_{\alpha_i}) - \hat t_{\alpha_i} =q(\gamma_-, \hat x_{\alpha_i}) - \hat t_{\alpha_i}$ for $i=1,2$.

Noting that $\gamma_-\neq\gamma_+$ and $q(\gamma_-,\hat x_{\alpha_i})=q(\gamma_+,\hat x_{\alpha_i})$, we conclude that $q_0(\hat x_{\alpha_i}) = q(\gamma_i,\hat x_{\alpha_i})\leq \hat t_{\alpha_i}$ and $q_1(\hat x_{\alpha_i}) = 0$. Thus, $(\hat x_{\alpha_i},\hat t_{\alpha_i})\in\Sequal$ for $i=1,2$. We conclude $(\hat x,\hat t)\in\conv(\Sequal)$.\qedhere
\end{enumerate}
\end{proof}

\begin{remark}
The proof of Lemma~\ref{lem:pd:Ssupset} suggests a simple rounding scheme from the convex relaxation to the original nonconvex problem: given $\hat x\in\R^n$, let $d$ be an eigenvector of eigenvalue zero for either $A(\gamma_{\pm})$ (depending on the sign of $q_1(\hat x)$) and move $\alpha\geq 0$ units in the direction of either $\pm d$ (depending on the sign of $e$ defined in the proof) until $q_1(\hat x\pm \alpha d) = 0$. This rounding scheme guarantees that $q_0(\hat x\pm\alpha d)\leq \max\set{q(\gamma_-, \hat x), q(\gamma_+, \hat x)}$.
\end{remark}

We are now ready to prove Theorem \ref{thm:pd_implies_conv_hull}.
\begin{proof}
[Proof of Theorem \ref{thm:pd_implies_conv_hull}]
Lemmas~\ref{lem:pd:Ssubset}~and~\ref{lem:pd:Ssupset} together imply
\begin{align*}
\cS(\gamma_-)\cap\cS(\gamma_+) \subseteq \conv(\Sequal) \subseteq \conv(\cS) \subseteq \cS(\gamma_-)\cap \cS(\gamma_+).
\end{align*}
Hence, we deduce that equality holds throughout the chain of inclusions.

In particular, the GTRS  \eqref{eq:GTRSepi} can be rewritten
\begin{align*}
\inf_{(x,t)\in\R^{n+1}}\set{t:~ (x,t) \in \cS} &=\inf_{(x,t)\in\R^{n+1}}\set{t:~ (x,t) \in \conv(\cS)}\\
&= \inf_{(x,t)\in\R^{n+1}}\set{t:~ (x,t) \in \cS(\gamma_-)\cap\cS(\gamma_+)}\\
&= \inf_{(x,t)\in\R^{n+1}}\set{t:~ \begin{array}
	{l}
	q(\gamma_-,x)\leq t\\
	q(\gamma_+,x)\leq t
\end{array}}\\
&= \inf_{x\in\R^n} \max\set{q(\gamma_-,x), q(\gamma_+,x)}.
\end{align*}
It remains to prove that the minimum is achieved in each of the formulations of the GTRS above.
It suffices to show that the minimum is achieved in the last formulation. Note $q(\gamma_-,x)$ and $q(\gamma_+,x)$ are both continuous functions of $x$, hence $\max\set{q(\gamma_-,x), q(\gamma_+,x)}$ is continuous. Next, taking $u\coloneqq \max\set{c(\gamma_-),c(\gamma_+)}$ we have that $u$ is an upper bound on the optimal value. Moreover, because $\gamma^*\in(\gamma_-,\gamma_+)$, we can lower bound $\max\set{q(\gamma_-,x),q(\gamma_+,x)}$, by $q(\gamma^*,x)$. Consequently, it suffices to replace the feasible domain $\R^n$ in the last formulation with the set
\begin{align*}
\set{x\in\R^n:~ q(\gamma^*,x)\leq u}.
\end{align*}
This set is bounded as $A(\gamma^*)\succ 0$ and it is closed as it is the inverse image of $(-\infty,u]$ under a continuous map. Recalling that a continuous function on a compact set achieves its minimum concludes the proof.
\end{proof}

We next provide a numerical example illustrating Theorem~\ref{thm:pd_implies_conv_hull}.
\begin{example}
\label{ex:convex_hull_result}
Define the homogeneous quadratic functions $q_i(x)\coloneqq x^\top A_i x$ for $i=0,1$, where
\begin{align*}
A_0 \coloneqq \begin{pmatrix}
	1 & 2\\ 2& 1
\end{pmatrix},\qquad
A_1 \coloneqq \begin{pmatrix}
	0 & -1\\ -1&0
\end{pmatrix}.
\end{align*}
As $\det(A_0) = -3$ and $\det(A_1) = -1$, the matrices $A_0$ and $A_1$ must both have negative eigenvalues. Furthermore, 
\begin{align*}
A(2) = A_0 + 2 A_1 = I \succ 0.
\end{align*}
Thus, Assumption~\ref{as:pd} is satisfied.

We now compute $\gamma_-$ and $\gamma_+$.
Note that as $A(\gamma)$ is a $2\by 2$ matrix, $A(\gamma)\succeq 0$ if and only if $\tr(A(\gamma))\geq 0$ and $\det(A(\gamma))\geq 0$. Note that
$\tr(A(\gamma))=2\geq 0$ is satisfied for all $\gamma$. We compute
\begin{align*}
\det(A(\gamma)) &= 1 - (2-\gamma)^2.
\end{align*}
This quantity is nonnegative if and only if $\abs{2-\gamma}\leq 1$. Thus $\gamma_-= 1$ and $\gamma_+=3$. 
Theorem~\ref{thm:pd_implies_conv_hull} then implies
\begin{align*}
\conv\left(\set{(x,t)\in\R^3 :\, \begin{array}
	{l}
	x_1^2 + 4x_1x_2 +x_2^2 \leq t\\
	-2x_1x_2 \leq 0
\end{array}}\right)&= 
\set{(x,t)\in\R^3 :\, \begin{array}
	{l}
	(x_1+x_2)^2 \leq t\\
	(x_1-x_2)^2 \leq t
\end{array}}.
\end{align*}

We plot the corresponding sets $\cS$ and $\cS(\gamma_-)\cap \cS(\gamma_+)$ in Figure~\ref{fig:example_convex_hull}.
\begin{figure}
  \centering
    \includegraphics[width=0.4\textwidth]{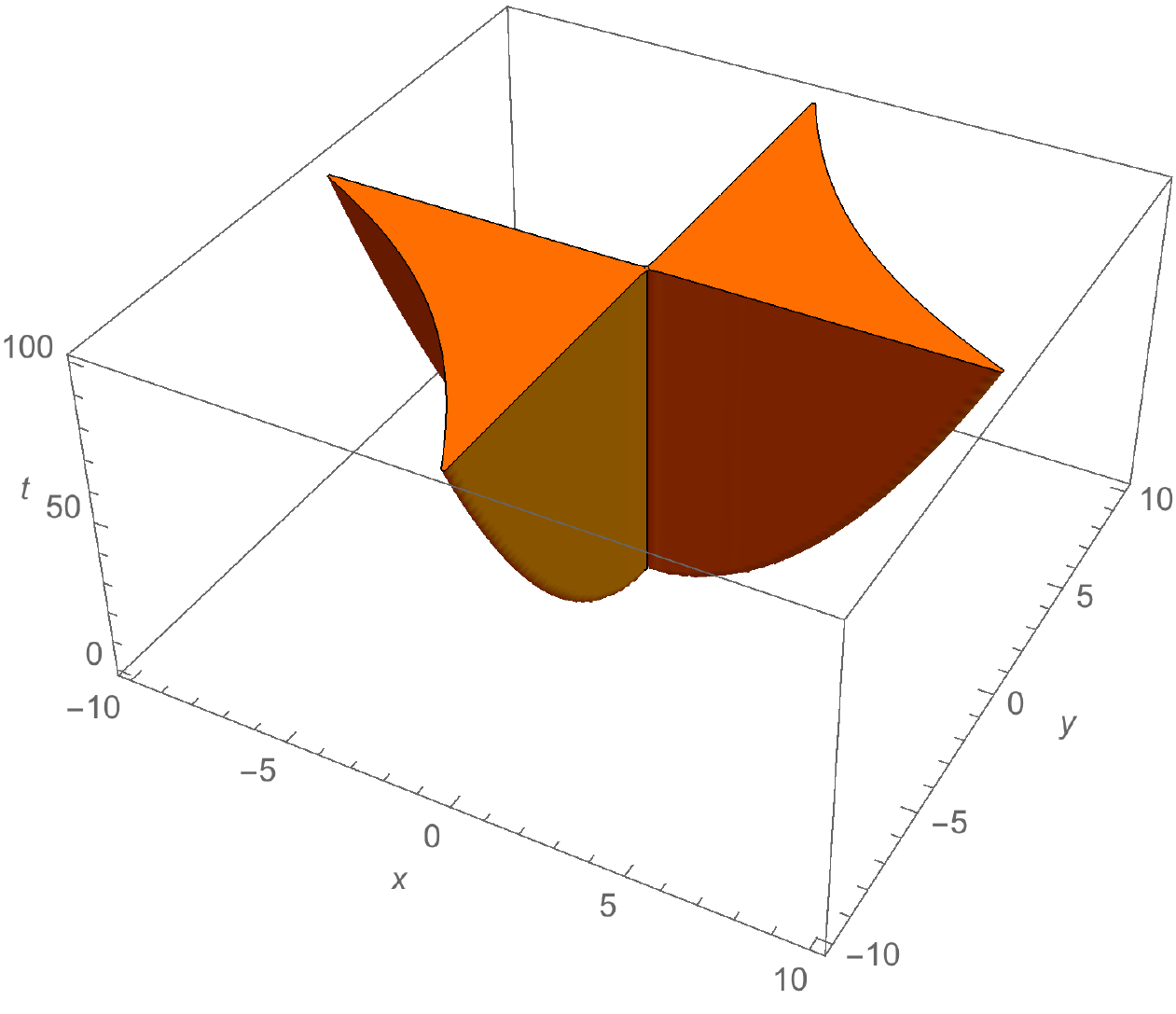}\qquad
    \includegraphics[width=0.4\textwidth]{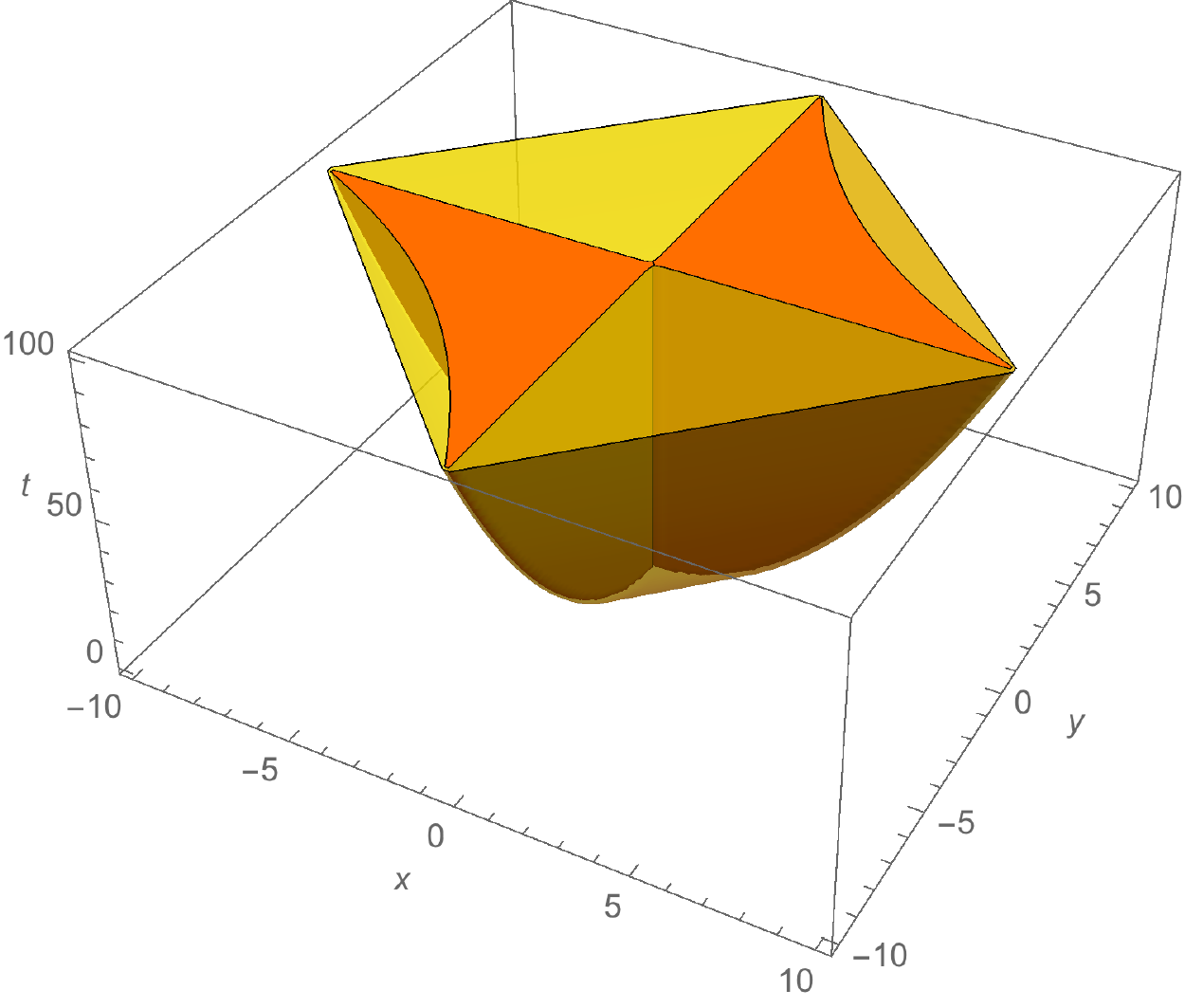}
	\caption{The sets $\cS$ (in orange) and $\cS(\gamma_-)\cap\cS(\gamma_+)$ (in yellow) from Example~\ref{ex:convex_hull_result}}
	\label{fig:example_convex_hull}
\end{figure}
\end{example}

\subsection{Proof of Theorem \ref{thm:psd_implies_conv_hull}}
\label{sec:proof_psd}
Next, we prove Theorem \ref{thm:psd_implies_conv_hull} using a limiting argument and reducing it to Theorem \ref{thm:pd_implies_conv_hull}.

\begin{lemma}\label{lem:psd:Ssubset}
Suppose $\Gamma$ is nonempty and write $\Gamma = [\gamma_-,\gamma_+]$. Then, ${\clconv}(\cS)\subseteq \cS(\gamma_-) \cap \cS(\gamma_+)$.
\end{lemma}
\begin{proof}
Note that $\cS = \bigcap_{\gamma\geq 0} \cS(\gamma) \subseteq \cS(\gamma_-)\cap \cS(\gamma_+)$. Containment then follows by taking the closed convex hull of both sides and noting that both $\cS(\gamma_-)$ and $\cS(\gamma_+)$ are closed and convex  by Lemma \ref{lem:psd_implies_convex_closed}.
\end{proof}

\begin{lemma}\label{lem:psd:Ssupset}
Under Assumption~\ref{as:psd}, we have that $\cS(\gamma_-)\cap \cS(\gamma_+)\subseteq {\clconv}(\Sequal)$.
\end{lemma}
\begin{proof}
Let $(\hat x,\hat t)\in \cS(\gamma_-)\cap \cS(\gamma_+)$. It suffices to show that $(\hat x,\hat t+\epsilon)\in\conv(\Sequal)$ for all $\epsilon>0$.

We will perturb $A_0$ slightly to create a new GTRS instance.
Let $\delta>0$ to be picked later. Define $A_0' = A_0 + \delta I_n$ and let all remaining data be unchanged, i.e.,
\begin{align*}
q_0'(x) &\coloneqq x^\top A_0' x + 2b_0'^\top x + c_0' \coloneqq x^\top (A_0+\delta I_n) x + 2b_0^\top x + c_0\\
q_1'(x) &\coloneqq x^\top A_1' x + 2b_1'^\top x + c_1' \coloneqq x^\top A_1 x + 2b_1^\top x + c_1.
\end{align*}
We will denote all quantities related to the perturbed system with an apostrophe. 

We claim that it suffices to show that there exists a $\delta>0$ small enough such that the GTRS defined by $q_0'$ and $q_1'$ satisfies Assumption~\ref{as:pd} and $(\hat x,\hat t+\epsilon)\in \cS'(\gamma_-')\cap \cS'(\gamma_+')$ . Indeed, suppose this is the case. Note that for any $x\in\R^n$, we have $q_1(x)= q'_1(x)$ and $q_0(x)\leq q'_0(x)$. Hence, $\Sequal'\subseteq\Sequal$ and $\conv(\Sequal')\subseteq\conv(\Sequal)$.
Then applying Theorem~\ref{thm:pd_implies_conv_hull} gives $(\hat x,\hat t+\epsilon)\in\cS'(\gamma_-')\cap \cS'(\gamma_+') = \conv(\Sequal')\subseteq\conv(\Sequal)$ as desired.

We pick $\delta>0$ small enough such that
\begin{align*}
\lambda_{\min}(A_0')<0,\quad
\delta\norm{\hat x}^2 \leq \frac{\epsilon}{2},\quad
\abs{\gamma'_+ - \gamma_+}\abs{q_1(\hat x)} \leq \frac{\epsilon}{2},
\quad
\abs{\gamma'_- - \gamma_-}\abs{q_1(\hat x)} \leq \frac{\epsilon}{2}.
\end{align*}
This is possible as the expression on the left of each inequality is continuous in $\delta$ and is strictly satisfied if $\delta=0$.
Then, noting that $A'(\gamma^*) = A(\gamma^*) + \delta I_n\succ 0$, we have that the GTRS defined by $q_0'$ and $q_1'$ satisfies Assumption~\ref{as:pd}.

It remains to show that $q'(\gamma'_+,\hat x)\leq (\hat t+\epsilon)$ and $q'(\gamma'_-,\hat x)\leq (\hat t+\epsilon)$. We compute
\begin{align*}
q'(\gamma'_+, \hat x) - (\hat t+\epsilon)
&= q'(\gamma_+,\hat x) - (\hat t+\epsilon) + (\gamma'_+ - \gamma_+)q_1(\hat x)\\
&\leq q(\gamma_+,\hat x)+\delta \norm{\hat x}^2 - (\hat t+\epsilon) +\abs{\gamma'_+ - \gamma_+}\abs{q_1(\hat x)}\\
&\leq q(\gamma_+,\hat x) - \hat t\\
&\leq 0.
\end{align*}
The first inequality follows by noting $q'(\gamma,x) = q(\gamma,x) + \delta\norm{x}^2$, the second inequality follows from our assumptions on $\delta$, and the third line follows from the assumption that $(\hat x, \hat t)\in \cS(\gamma_+)$. 
A similar calculation shows $q'(\gamma'_-, \hat x)\leq (t+\epsilon)$. This concludes the proof.
\end{proof}

We are now ready to prove Theorem \ref{thm:psd_implies_conv_hull}.
\begin{proof}
[Proof of Theorem \ref{thm:psd_implies_conv_hull}]
Lemmas~\ref{lem:psd:Ssubset}~and~\ref{lem:psd:Ssupset} together imply
\begin{align*}
\cS(\gamma_-)\cap\cS(\gamma_+) \subseteq {\clconv}(\Sequal) \subseteq {\clconv}(\cS) \subseteq \cS(\gamma_-)\cap\cS(\gamma_+).
\end{align*}
Hence, we deduce that equality holds throughout the chain of inclusions.

In particular, the GTRS \eqref{eq:GTRSepi} can be rewritten
\begin{align*}
\inf_{(x,t)\in\R^{n+1}}\set{t:~ (x,t)\in \cS}
&= \inf_{(x,t)\in\R^{n+1}}\set{t:~ (x,t)\in {\clconv}(\cS)}\\
&= \inf_{(x,t)\in\R^{n+1}}\set{t:~ (x,t)\in \cS(\gamma_-)\cap \cS(\gamma_+)}\\
&= \inf_{(x,t)\in\R^{n+1}}\set{t:~ \begin{array}
	{l}
	q(\gamma_-,x)\leq t\\
	q(\gamma_+,x)\leq t
\end{array}}\\
&= \inf_{x\in\R^n}\max\set{q(\gamma_-,x),q(\gamma_+,x)}.\qedhere
\end{align*}
\end{proof}

\subsection{Removing the nonconvex assumptions}

As part of our Assumptions~\ref{as:pd}~and~\ref{as:psd}, we assume that $A_0$ and $A_1$ both have negative eigenvalues, i.e., that both $q_0$ and $q_1$ are nonconvex.
These assumptions are made for ease of presentation and to highlight the novel contributions of this work.
Indeed, the proofs of Theorems~\ref{thm:pd_implies_conv_hull}~and~\ref{thm:psd_implies_conv_hull} can be modified to additionally cover all four cases of convex/nonconvex objective and constraint functions. We remark that the resulting theorem statement for the case of a nonconvex objective function and a strongly convex constraint function coincides with that of \citet{Ho-NguyenKK2017}.

In this section we record more general versions Theorems~\ref{thm:pd_implies_conv_hull}~and~\ref{thm:psd_implies_conv_hull}. Their proofs are completely analogous to the original proofs and are deferred to Appendix~\ref{app:general_conv_hull_proofs}.

\begin{restatable}{theorem}{pdGeneralConvHull}
\label{thm:pd_implies_conv_hull_general}
Suppose there exists $\gamma^*\geq 0$ such that $A(\gamma^*)\succ 0$. Consider the closed nonempty interval $\Gamma\coloneqq \set{\gamma\in\R_+:\, A(\gamma)\succeq 0}$. Let $\gamma_-$ denote its leftmost endpoint.
\begin{itemize}
	\item If $\Gamma$ is bounded above, let $\gamma_+$ denote its rightmost endpoint. Then,
	\begin{align*}
	\conv(\cS) = \cS(\gamma_-)\cap \cS(\gamma_+).
	\end{align*}
	In particular, we have $\min_{x\in\R^n}\set{q_0(x):\, q_1(x)\leq 0} = \min_{x\in\R^n} \max\set{q(\gamma_-, x),\, q(\gamma_+,x)}$.
	\item If $\Gamma$ is not bounded above, then $q_1(x)$ is convex and
	\begin{align*}
	\conv(\cS) = \cS(\gamma_-)\cap \set{(x,t)\in\R^{n+1}:\, q_1(x)\leq 0}.
	\end{align*}
	In particular, we have $\min_{x\in\R^n}\set{q_0(x):\, q_1(x)\leq 0} = \min_{x\in\R^n} \set{q(\gamma_-,x):\, q_1(x)\leq 0}$.
\end{itemize}
\end{restatable}

\begin{restatable}{theorem}{psdGeneralConvHull}
\label{thm:psd_implies_conv_hull_general}
Suppose there exists $\gamma^*\geq 0$ such that $A(\gamma^*)\succeq 0$. Consider the closed nonempty interval $\Gamma\coloneqq \set{\gamma\in\R_+:\, A(\gamma)\succeq 0}$. Let $\gamma_-$ denote its leftmost endpoint.
\begin{itemize}
	\item If $\Gamma$ is bounded above, let $\gamma_+$ denote its rightmost endpoint. Then,
	\begin{align*}
	\overline{\conv}(\cS) = \cS(\gamma_-)\cap \cS(\gamma_+).
	\end{align*}
	In particular, $\inf_{x\in\R^n}\set{q_0(x):\, q_1(x)\leq 0} = \inf_{x\in\R^n} \max\set{q(\gamma_-, x),\, q(\gamma_+,x)}$.
	\item If $\Gamma$ is not bounded above, then $q_1(x)$ is convex and
	\begin{align*}
	\overline{\conv}(\cS) = \cS(\gamma_-)\cap \set{(x,t)\in\R^{n+1}:\, q_1(x)\leq 0}.
	\end{align*}
	In particular, $\inf_{x\in\R^n}\set{q_0(x):\, q_1(x)\leq 0} = \inf_{x\in\R^n} \set{q(\gamma_-,x):\, q_1(x)\leq 0}$.
\end{itemize}
\end{restatable}
These results admit further nontrivial generalizations involving multiple quadratics; we refer the interested readers to our follow up work~\cite{wang2019tightness}.

\begin{remark}\label{rem:convComparisonY}
\citet{yildiran2009convex} proves a convex hull result for a set defined by two \emph{strict} quadratic constraints.
\citet{modaresi2017convex} then show that given a particular topological assumption, that the appropriate closed versions of \citet{yildiran2009convex}'s results also hold.
We discuss these results in the context of the convex hull results we have presented thus far. Given $q_0$ and $q_1$ we will consider the quadratic functions $q_0(x)-t$ and $q_1(x)$ in the variables $(x,t)$. As \cite{yildiran2009convex} works with homogeneous quadratics, we introduce an extra variable to get homogeneous quadratic forms. Define
\begin{align*}
Q_0 \coloneqq \begin{pmatrix}
	A_0 & 0 & b_0 \\
	0^\top & 0 & -1/2\\
	b_0^\top & -1/2 & c_0
\end{pmatrix}
,\quad
Q_1 \coloneqq \begin{pmatrix}
	A_1 & 0 & b_1 \\
	0^\top & 0 & 0\\
	b_1^\top & 0 & c_1
\end{pmatrix}
,\quad
Q(\gamma) \coloneqq \begin{pmatrix}
	A(\gamma) & 0 & b(\gamma)\\
	0^\top & 0 & -1/2\\
	b(\gamma)^\top & -1/2 & c(\gamma)
\end{pmatrix}.
\end{align*}
\citet{yildiran2009convex} uses the aggregation weights $\gamma$ where $Q(\gamma)$ has exactly one negative eigenvalue. Note that for all $\gamma\geq 0$, the lower right $2\by 2$ block of $Q(\gamma)$ is invertible. Thus, we may take the Schur complement of this block in $Q(\gamma)$:
\begin{align*}
Q(\gamma)/\begin{pmatrix}
	0 & -1/2 \\ -1/2 & c(\gamma)
\end{pmatrix} = A(\gamma) - \begin{pmatrix}
	0 & b(\gamma)
\end{pmatrix}
\begin{pmatrix}
	0 & -1/2 \\ -1/2 & c(\gamma)
\end{pmatrix}^{-1}
\begin{pmatrix}
	0^\top \\ b(\gamma)^\top
\end{pmatrix} = A(\gamma).
\end{align*}
Recall that Schur complements preserve inertia. In other words, $Q(\gamma)$ and
\begin{align*}
\begin{pmatrix}
	A(\gamma) &  & \\
	 & 0 & -1/2\\
	 & -1/2 & c(\gamma).
\end{pmatrix}
\end{align*}
have the same number of negative eigenvalues. Noting that the lower right $2\by 2$ block has exactly one negative eigenvalue, we conclude that $Q(\gamma)$ has exactly one negative eigenvalue if and only if $A(\gamma)\succeq 0$.
The result presented by \citet{yildiran2009convex} then implies
\begin{align*}
\conv\left(\set{(x,t):\, \begin{array}
	{l}
	q_0(x)<t\\
	q_1(x)<0
\end{array}}\right) = \set{(x,t):\, q(\gamma_-,x)<t}\cap\set{(x,t):\, q(\gamma_+,x)<t}
\end{align*}
when $\gamma_+$ exists and
\begin{align*}
\conv\left(\set{(x,t):\, \begin{array}
	{l}
	q_0(x)<t\\
	q_1(x)<0
\end{array}}\right) = \set{(x,t):\, q(\gamma_-,x)<t}\cap\set{(x,t):\, q_1(x)<0}
\end{align*}
otherwise.

One can then verify the topological assumption of \citet{modaresi2017convex}, namely that $\cS \subseteq\overline{\inter(\cS)}$ the closure of the interior of $\cS$. Thus, combining these two results gives an alternate proof of Theorems~\ref{thm:pd_implies_conv_hull_general}~and~\ref{thm:psd_implies_conv_hull_general}.

We believe our analysis is simpler and more direct. In particular, our analysis takes advantage of the epigraph structure present in our sets and immediately implies a rounding procedure via Lemma~\ref{lem:pd:Ssupset}. In addition, our results are more refined when Assumption~\ref{as:pd}~or~\ref{as:psd} hold as we can also characterize the (closed) convex hull of the set $\Sequal$ and show that it is equal to that of $\cS$. This particular distinction between $\Sequal$ and $\cS$ has a number of interesting implications in equality-, interval-, or hollow-constrained GTRS, and we discuss these results in the following section.
\end{remark} 

\section{Nonintersecting constraints}
\label{sec:nonintersecting_constraints}

There have been a number of works considering interval-, equality-, or hollow-constrained variants of the GTRS~\cite{BenTalTeboulle1996,BenTalDenHertog2014,jiang2018socp,jiang2019novel,PongWolkowicz2014,salahi2018efficient,more1993generalizations,SternWolkowicz1995,	yang2018quadratic} (see \cite[Section 3.3]{Ho-NguyenKK2017} and references therein for extensions of the TRS and their applications).
In this section, we extend our (closed) convex hull results in the presence of a general nonintersecting constraint. This allows us to handle multiple variants of the GTRS simultaneously.

Specifically, we will impose an additional requirement $x\in\Omega$. The new form of the GTRS will be 
\begin{align*}
\inf_{x\in\R^n}\set{q_0(x):~ \begin{array}
	{l}
	q_1(x) \leq 0\\
	x\in \Omega
\end{array}} = \inf_{(x,t)\in\R^{n+1}}\set{t:~ \begin{array}
	{l}
	q_0(x)\leq t\\
	q_1(x) \leq 0\\
	x\in \Omega
\end{array}}.
\end{align*}
Let $\cS_\Omega$ denote the set of feasible points $(x,t)$, i.e.,
\begin{align*}
\cS_\Omega \coloneqq 
\set{(x,t)\in\R^{n+1} :~ \begin{array}
	{l}
	q_0(x) \leq t\\
	q_1(x) \leq 0\\
	x\in\Omega
\end{array}}.
\end{align*}

We will assume that $\Omega\subseteq\R^n$ satisfies the following \emph{nonintersecting} condition.
\begin{assumption}\label{as:nonintersecting}
The set $\Omega\subseteq\R^n$ satisfies $\set{x\in\R^n :~ q_1(x) = 0}\subseteq\Omega$.
\end{assumption}
The following two corollaries to Theorems \ref{thm:pd_implies_conv_hull} and \ref{thm:psd_implies_conv_hull} follow immediately by noting that $\Sequal\subseteq \cS_\Omega \subseteq \cS$ holds under Assumption~\ref{as:nonintersecting}.
\begin{corollary}
Suppose Assumptions~\ref{as:pd}~and~\ref{as:nonintersecting} hold. Then,
\begin{align*}
\conv(\cS_\Omega)=\cS(\gamma_-)\cap \cS(\gamma_+).
\end{align*}
\end{corollary}

\begin{proof}
Under Assumptions~\ref{as:pd}~and~\ref{as:nonintersecting}, we get the following chain of inclusions
\begin{align*}
\conv(\cS_\Omega) \subseteq \conv(\cS) = \conv(\Sequal) \subseteq \conv(\cS_\Omega),
\end{align*}
where the first subset relation follows $\cS_\Omega \subseteq \cS$ (by definition of the set $\cS_\Omega$), the equality relation follows from Theorem~\ref{thm:pd_implies_conv_hull}, and the last subset relation follows from $\Sequal\subseteq \cS_\Omega$ (by  Assumption~\ref{as:nonintersecting}). 
We conclude that $\conv(\cS_\Omega)=\conv(\cS)$. By Theorem~\ref{thm:pd_implies_conv_hull}, we know that $\conv(\cS) = \cS(\gamma_-)\cap \cS(\gamma_+)$.
\end{proof}

\begin{corollary}
Suppose Assumptions~\ref{as:psd}~and~\ref{as:nonintersecting} hold. Then,
\begin{align*}
\overline{\conv}(\cS_\Omega)=\cS(\gamma_-)\cap \cS(\gamma_+).
\end{align*}
\end{corollary}

\begin{proof}
Applying Assumptions~\ref{as:psd}~and~\ref{as:nonintersecting} and Theorem~\ref{thm:psd_implies_conv_hull}, we get the following chain of inclusions
\begin{align*}
\overline{\conv}(\cS_\Omega) \subseteq \overline{\conv}(\cS) = \overline{\conv}(\Sequal) \subseteq \overline{\conv}(\cS_\Omega).
\end{align*}
We conclude that $\overline{\conv}(\cS_\Omega)=\overline{\conv}(\cS)$. By Theorem~\ref{thm:psd_implies_conv_hull}, we know that $\overline{\conv}(\cS) = \cS(\gamma_-)\cap \cS(\gamma_+)$.
\end{proof}

\begin{remark}
These two corollaries show that nonintersecting constraints in the GTRS may be ignored. Consider for example the interval-constrained GTRS. Define
\begin{align*}
\Omega \coloneqq \set{x\in\R^n:~ q_1(x) \geq -1}.
\end{align*}
Then, clearly Assumption~\ref{as:nonintersecting} is satisfied. 
Under Assumption~\ref{as:psd}, we have 
\begin{align*}
\inf_{x\in\R^n}\set{q_0(x) :\, -1\leq q_1(x)\leq 0}
&= \inf_{(x,t)\in\R^{n+1}}\set{t:~(x,t)\in\cS_\Omega}\\
&= \inf_{(x,t)\in\R^{n+1}}\set{t:~ (x,t)\in \overline{\conv}(\cS_\Omega)}\\
&= \inf_{(x,t)\in\R^{n+1}}\set{t:~ (x,t)\in \cS(\gamma_-)\cap \cS(\gamma_+)}\\
&= \inf_{x\in\R^n} \max\set{q(\gamma_-,x),q(\gamma_+,x)}.
\end{align*}
Thus, the value of the interval-constrained GTRS is the same as the GTRS under Assumption~\ref{as:psd}. Similarly, the $\Omega$ sets arising from equality- or hollow-constrained GTRS also satisfy Assumption~\ref{as:nonintersecting}. Hence, under Assumption~\ref{as:psd}, the additional constraints in these variants of the GTRS can also be dropped.
\end{remark} 

\section{Solving the convex reformulation in linear time}
\label{sec:alg}

In this section we present algorithms, inspired by Theorem \ref{thm:pd_implies_conv_hull}, for approximately solving the GTRS.
Note that Theorem \ref{thm:pd_implies_conv_hull} gives a tight convex reformulation of the GTRS: under Assumption \ref{as:pd},
\begin{align*}
\Opt \coloneqq \min_{x\in\R^n} \set{q_0(x):\,q_1(x)\leq 0} = \min_{x\in\R^n}\max\set{q(\gamma_-,x),q(\gamma_+,x)}.
\end{align*}
Then given a solution to the convex reformulation on the right, Lemma~\ref{lem:pd:Ssupset} gives a rounding scheme to recover a solution to the original GTRS on the left.

In order to establish an explicit running time of an algorithm based on the above idea, we must carefully handle a number of numerical issues.
In practice, we cannot expect to compute $\gamma_\pm$ exactly.
Instead, we will show how to compute estimates $\tilde\gamma_\pm$ of $\gamma_\pm$ up to some accuracy $\delta$.
We will take care to pick $\tilde \gamma_\pm$ satisfying the relation $[\tilde \gamma_-,\tilde \gamma_+]\subseteq [\gamma_-,\gamma_+]$ so that the quadratic forms defined by $A(\tilde\gamma_-)$ and $A(\tilde\gamma_+)$ are convex.
Based on the estimates $\tilde\gamma_\pm$, we will then formulate and solve the convex optimization problem
\[
\widetilde\Opt \coloneqq \min_{x\in\R^n}\max \set{q(\tilde\gamma_-,x),q(\tilde\gamma_+,x)}.
\]
Finally, given an (approximate) solution to the convex problem $\widetilde\Opt$, Lemma~\ref{lem:pd:Ssupset} tells us how to construct a solution to the original nonconvex GTRS using specific eigenvectors. Again, we will need to handle numerical issues that arise from not being able to compute these eigenvectors exactly.

Throughout this section, we will work under the following assumption.
\begin{assumption}
	\label{as:alg}
	\leavevmode
	\begin{itemize}
		\item There exists some $\gamma^*\geq0$ such that $A(\gamma^*)\succ 0$,
		\item $\norm{A_0},\norm{A_1},\norm{b_0},\norm{b_1},\abs{c_1}\leq 1$.\qedhere
	\end{itemize}
\end{assumption}

\begin{remark}
Note that the first part of Assumption~\ref{as:alg} is simply Assumption~\ref{as:pd}. We make this assumption so that we may use the convex reformulation guaranteed by Theorem~\ref{thm:pd_implies_conv_hull}. 
Assumption~\ref{as:pd} is commonly used in GTRS algorithms; see e.g., \citet[Assumption 2.3]{jiang2020linear} and the discussion following it.  
The second part of Assumption~\ref{as:alg} can be achieved for an arbitrary pair $q_0$ and $q_1$ by simply scaling each quadratic by a positive scalar. Note that any optimal (respectively feasible) solution remains optimal (respectively feasible) when $q_0$ (respectively $q_1$) is scaled by a positive scalar.
\end{remark}

We will analyze the running time of our algorithm in terms of $N$, the number of nonzero entries in $A_0$ and $A_1$, $\epsilon$, the additive error, $p$, the failure probability, and $n$, the dimension.
In addition, the running time of our algorithm depends on certain regularity parameters of the pair $q_0$ and $q_1$ defined below.
\begin{definition}
\label{def:regularity}
Let $q_0$, $q_1$ satisfy Assumption \ref{as:alg}. Define
\begin{align*}
\zeta^*\coloneqq \max\set{1,\gamma_+},
~~\text{ and }~~
\xi^*\coloneqq \min\set{1, \max_{\gamma\geq 0}\lambda_{\min}(A(\gamma))}.
\end{align*}
We say that $q_0$ and $q_1$ are $(\xi,\zeta)$ regular if $0<\xi\leq\xi^*$ and $\zeta\geq\zeta^*$.
Define $\kappa^* = \zeta^*/\xi^*$. When $(\xi,\zeta)$ are clear from context we will write $\kappa\coloneqq \zeta/\xi$.
\end{definition}

In our analysis, we will frequently use the inequalities $\kappa,\zeta,\xi^{-1}\geq 1$,  which for example imply $\kappa^2\geq \kappa$ and $1+\kappa \leq 2\kappa$, and the inequalities $\gamma_-\leq\gamma_+\leq \zeta$, which for example under Assumption~\ref{as:alg} imply $\norm{A(\gamma_+)}\leq 1+\zeta\leq 2\zeta$.

\begin{remark}
\label{rem:regularity_comparison}
\citet{jiang2020linear} present a different linear-time algorithm for solving the GTRS. In their paper, they assume they are given a regularity parameter $\xi_{\text{JL}}$ as input. This parameter must satisfy $\xi_{\text{JL}}\leq\xi_{\text{JL}}^*$ where
\begin{align*}
\xi^*_{\text{JL}} &\coloneqq \min\set{1,-\lambda_{\min}(A_1), \max_{\mu\in(0,1]} \lambda_{\min}\left(\mu A_0 +(1-\mu) A_1 \right)}.
\end{align*}
We now discuss how our regularity parameters, $\xi^*_{\text{us}}$, $\zeta^*$, and $\kappa^*\coloneqq \frac{\zeta^*}{\xi^*_{\text{us}}}$ relate to $\xi^*_{\text{JL}}$. For simplicity, we will assume
\begin{align*}
\xi^*_{\text{us}} &=\max_{\gamma\geq 0} \set{\lambda_{\min} (A(\gamma))},\hspace{2em}
\zeta^* = \gamma_+,\\
\xi^*_{\text{JL}} &= \min\set{-\lambda_{\min}(A_1), \max_{\mu\in(0,1]} \lambda_{\min}\left(\mu A_0 +(1-\mu) A_1 \right)}.
\end{align*}

We claim $\zeta^* \leq (-\lambda_{\min}(A_1))^{-1}$. Indeed, let $x$ be a unit eigenvector corresponding to $\lambda_{\min}(A_1)$. Then, for any $\gamma>(-\lambda_{\min}(A_1))^{-1}$, we have
\begin{align*}
x^\top A(\gamma)x &= x^\top A_0 x + \gamma x^\top A_1 x \leq 1 + \gamma\lambda_{\min}(A_1)<0.
\end{align*}
The role played by the bound $\gamma_+\leq \zeta$ in our analysis is similar to the role of the bound $\xi_{\text{JL}}\leq -\lambda_{\min}(A_1)$ in the analysis presented by \citet{jiang2020linear}.

We claim that
\begin{align*}
\frac{1}{2\kappa^*}\leq \max_{\mu\in(0,1]}\lambda_{\min}(\mu A_0 + (1-\mu)A_1) \leq \xi^*_{\text{us}},
\end{align*}
and that the lower bound is sharp. 
Indeed, by performing the transformation $\mu = \frac{1}{1+\gamma}$, we can rewrite
\begin{align*}
\max_{\mu\in(0,1]}\lambda_{\min}\left(\mu A_0 +(1-\mu) A_1 \right)
&= \max_{\gamma\geq 0} \frac{1}{1+\gamma} \lambda_{\min}(A(\gamma)),
\end{align*}
which we can clearly bound above by $\xi^*_{\text{us}}$.
On the other hand, noting that any optimizer, $\gamma$, of the above problem must lie in $[0,\gamma_+]=[0,\zeta^*]$, we can lower bound
\begin{align*}
\max_{\gamma\geq 0} \frac{1}{1+\gamma} \lambda_{\min}(A(\gamma))
&\geq \frac{1}{1+\zeta^*} \max_{\gamma\geq 0} \lambda_{\min}(A(\gamma))= \frac{\xi^*_{\text{us}}}{1+\zeta^*}\geq \frac{1}{2\kappa^*}.
\end{align*}

We now construct a simple example for which the lower bound, $\xi^*_{\text{JL}} \geq \tfrac{1}{2\kappa^*}$, is sharp. Let $\alpha>0$ and define
\begin{align*}
A_0 = \Diag(1,1,-1),\hspace{2em}
A_1 = \Diag\left(1,-(1+\alpha)^{-1}, 1\right).
\end{align*}
It is simple to see that $\norm{A_0} = \norm{A_1} = 1$, $\xi^*_{\text{us}} = \tfrac{\alpha}{2+\alpha}$ and $\zeta^* = 1+\alpha$. In particular, $\kappa^* = \tfrac{2+3\alpha+\alpha^2}{\alpha}$. 
On the other hand, we can compute 
\begin{align*}
\xi^*_{\text{JL}}
&= \max_{\mu\in(0,1]} \min\set{\mu - (1-\mu)\alpha, \mu(-1 + 2\alpha) + (1-\mu) \alpha}= \frac{\alpha}{4+3\alpha}.
\end{align*}
Then, letting $\alpha\to 0$, we have $\kappa^* = \tfrac{2 + o(1)}{\alpha}$ and $\xi^*_{\text{JL}}= \tfrac{\alpha}{4 + o(1)}$.

In view of the (closed) convex hull results presented in Theorems \ref{thm:pd_implies_conv_hull} and \ref{thm:psd_implies_conv_hull}, we believe that the right notion of regularity should depend on the parameterization $A_0 + \gamma A_1$ as opposed to $\mu A_0 + (1-\mu) A_1$. We compare the running time of the algorithm presented by \citet{jiang2020linear} and the running time of our algorithms in Remark~\ref{rem:runtime_comparison}.
\end{remark}

We will assume that we have access to these regularity parameters within our algorithms.
\begin{assumption}
	\label{as:alg_regularity}
	Assume we have algorithmic access to a pair $(\xi,\zeta)$ such that $q_0$ and $q_1$ are $(\xi,\zeta)$-regular and a $\hat\gamma$ satisfying $\lambda_{\min}(A(\hat\gamma))\geq \xi$.
\end{assumption}
\begin{remark}
Assumption~\ref{as:alg_regularity} is quite reasonable.
Indeed, there are simple and efficient binary search schemes to find constant factor approximations of $\xi^*$ and $\zeta^*$ and a corresponding $\hat\gamma$.
We detail one such algorithm in Appendix \ref{app:estimating_regularity}.
We remark that a similar assumption is made by \citet{jiang2020linear}: they assume they are given access to $\xi_{\text{JL}}$ and present an algorithm for computing a corresponding $\hat\mu$ (see Remark~\ref{rem:regularity_comparison}).
Another algorithm for finding $\hat\gamma$ is presented by \citet{guo2009improved} in the language of matrix pencil definiteness.
\end{remark}

We now fix the accuracy\footnote{Our definition of accuracy is presented in \eqref{eq:tilde_gamma_guarantee}.} to which we will compute our estimates $\tilde\gamma_\pm$. Define
\begin{equation}\label{eq:delta}
\delta\coloneqq \frac{\epsilon}{72\kappa^2}.
\end{equation}

The framework for our approach is shown in Algorithm \ref{alg:GTRS_outer}.

\begin{algorithm}
\caption{$\text{ApproxConvex}(q_0, q_1, \xi,\zeta,\hat\gamma, \epsilon, p)$}
\label{alg:GTRS_outer}
Given $q_0$ and $q_1$ satisfying Assumption~\ref{as:alg}, $(\xi,\zeta)$ and $\hat\gamma$ satisfying Assumption~\ref{as:alg_regularity}, error parameter $0<\epsilon\leq \kappa^2\xi$, and failure probability $p>0$
\begin{enumerate}
	\item Pick $\delta$ as in \eqref{eq:delta}.
	\item Find $\tilde\gamma_-$ and $\tilde\gamma_+$ such that
	\begin{equation}
	\label{eq:tilde_gamma_guarantee}
	\tilde\gamma_-\in[\gamma_-,\gamma_- + \delta],\hspace{2em}
	\tilde\gamma_+\in[\gamma_+-\delta,\gamma_+],\hspace{2em}
	\lambda_{\min}(A(\tilde\gamma_\pm))\leq \delta/\kappa,
	\end{equation}
	with failure probability of at most $p$.
	\item Define $\widetilde\Opt \coloneqq \min_{x\in\R^n}\max \set{q(\tilde\gamma_-,x),q(\tilde\gamma_+,x)}$. Solve $\widetilde\Opt$ up to accuracy $\epsilon/2$.
	\item Output $\tilde\gamma_-$, $\tilde\gamma_+$, and the approximate optimizer $\tilde x$.
\end{enumerate}
\end{algorithm}
Note that by Definition~\ref{def:regularity}, we have $\kappa^2\xi=\zeta^2/\xi\geq1$. Thus the requirement $0<\epsilon\leq \kappa^2\xi$ in Algorithm \ref{alg:GTRS_outer} is not a practical issue: given $\epsilon>\kappa^2\xi$, we can simply run our algorithm with $\epsilon' = 1$ and return a solution with a better error guarantee.

This section is structured as follows. In Section \ref{subsec:alg_delta}, we prove that when $\delta$ is picked according to \eqref{eq:delta},
$\widetilde \Opt$ is within $\epsilon/2$ of $\Opt$. In Section \ref{subsec:alg_tilde_gamma} we show how to compute $\tilde\gamma_-$ and $\tilde\gamma_+$ to satisfy \eqref{eq:tilde_gamma_guarantee}. Then in Section \ref{subsec:alg_nesterov}, we present an algorithm due to \citet{nesterov2018lectures} and show that it can be used to efficiently solve for $\widetilde\Opt$ up to accuracy $\epsilon/2$.
At the end of Section \ref{subsec:alg_nesterov}, we present Theorem \ref{thm:alg_runtime},
which collects the results of the previous subsections and formally analyzes the runtime of Algorithm \ref{alg:GTRS_outer}.
In Section \ref{subsec:alg_rounding}, we give a rounding scheme for finding a solution to the original GTRS \eqref{eq:GTRS} given a solution to the convex reformulation. Finally, in Section \ref{subsec:alg_further_remarks}, we show that the running times of our algorithms can be significantly improved in situations where it is easy to compute $\gamma_\pm$ and zero eigenvectors of $A(\gamma_\pm)$.

\subsection{Perturbation analysis of the convex reformulation}
\label{subsec:alg_delta}

In this subsection, we show that the perturbed convex reformulation, $\widetilde\Opt$, approximates the true convex reformulation, $\Opt$, up to an additive error of $\epsilon/2$ when $\delta$ is picked as in \eqref{eq:delta}.
We will assume that step 2 of Algorithm~\ref{alg:GTRS_outer} is successful, i.e., we have $\tilde\gamma_\pm$ satisfying \eqref{eq:tilde_gamma_guarantee}.

Recall the definition of $\delta$ in \eqref{eq:delta}.
As we require $\epsilon \leq \kappa^2\xi$, we will have
\begin{align*}
\delta \coloneqq \frac{\epsilon}{72\kappa^2} \leq \frac{\xi}{72} < \xi.
\end{align*}
It is easy to see that $\lambda_{\min}(A(\gamma))$ is a $1$-Lipschitz function in $\gamma$. Then recalling that $\lambda_{\min}(A(\gamma_\pm)) = 0$ and $\lambda_{\min}(A(\hat\gamma))\geq \xi$, we deduce the containment $\hat\gamma\in(\gamma_-+\delta,\gamma_+-\delta)$. 
This, along with \eqref{eq:tilde_gamma_guarantee}, implies
\begin{align}
\label{eq:gamma_ordering}
\hat\gamma\in (\tilde\gamma_-,\tilde\gamma_+)\subseteq[\gamma_-,\gamma_+],
\qquad
\tilde\gamma_- \in[\gamma_-,\gamma_-+\delta],
\qquad
\tilde\gamma_+ \in[\gamma_+-\delta,\gamma_+].
\end{align}

Recall the perturbed reformulation
\begin{align*}
\widetilde\Opt \coloneqq \min_{x\in\R^n} \max\set{q(\tilde\gamma_-,x),q(\tilde\gamma_+,x)}.
\end{align*}

For notational convenience, let $f(x)\coloneqq \max\set{q(\gamma_-,x),q(\gamma_+,x)}$ and let $\tilde f(x) \coloneqq \max\set{q(\tilde\gamma_-,x),q(\tilde\gamma_+,x)}$.
Let $x^*$ and $\tilde x^*$ denote optimizers of $\Opt$ and $\widetilde\Opt$ respectively.

\begin{lemma}\label{lem:ubOpt}
For any fixed $x\in\R^n$, we have $\tilde f(x)\leq f(x)$. In particular, $\widetilde\Opt \leq \Opt$.
\end{lemma}
\begin{proof}
Note that $q(\gamma,x)$ is a linear function in $\gamma$. Hence, for any fixed $x\in\R^n$, the containment $[\tilde\gamma_-,\tilde\gamma_+]\subseteq[\gamma_-,\gamma_+]$ implies $\tilde f(x) \leq f(x)$.
We deduce
\begin{align*}
\widetilde\Opt &\leq \tilde f(x^*)\leq f(x^*)= \Opt.\qedhere
\end{align*}
\end{proof}

To show $\widetilde\Opt \geq \Opt -\epsilon/2$, we will show that $x^*$ and $\tilde x^*$ lie in a ball of bounded radius and that $\tilde f$ approximates $f$ uniformly on this ball.

\begin{lemma}
\label{lem:locating_opt}
Let $x^*$ and $\tilde x^*$ be optimizers of $\Opt$ and $\widetilde\Opt$ respectively. Then
$x^*,\tilde x^*\in B(0,5\kappa)$.
\end{lemma}
\begin{proof}
By picking the feasible solution $0\in\R^n$ and Lemma~\ref{lem:ubOpt}, we have a trivial upper bound on $\widetilde\Opt$ and $\Opt$:
\begin{align}
\label{eq:locating_opt_upper_bound}
\widetilde\Opt \leq \Opt \leq \max\set{q(\gamma_-, 0),q(\gamma_+,0)} = \max\set{c(\gamma_-),c(\gamma_+)}.
\end{align}
By the first part of \eqref{eq:gamma_ordering}, we have
\[
f(x) \geq \tilde f(x)
\geq q(\hat\gamma,x)
\geq \xi\norm{x}^2 + 2b(\hat\gamma)^\top x + c(\hat\gamma),
\]
where the last inequality follows from the assumption that $\lambda_{\min}(A(\hat\gamma))\geq \xi$.
Then,
\begin{align*}
x^*,\tilde x^* &\in \set{x\in\R^n :\, \xi\norm{x}^2 + 2b(\hat\gamma)^\top x + c(\hat\gamma) \leq \max\set{c(\gamma_-),c(\gamma_+)}}\\
&\subseteq \set{x\in\R^n :\, \xi \norm{x}^2 + 2b(\hat\gamma)^\top  x\leq \zeta}.
\end{align*}
The last relation holds since $\max\set{c(\gamma_-)-c(\hat\gamma),c(\gamma_+)-c(\hat\gamma)} = \max\{(\gamma_- -\hat\gamma)c_1, (\gamma_+ -\hat\gamma)c_1\}\leq |c_1| \gamma_+ \leq \zeta$.
Then, by completing the square
\begin{align*}
x^*,\tilde x^*  &\in B\left(-b(\hat\gamma)\xi^{-1}, \sqrt{\norm{b(\hat\gamma)}^2\xi^{-2} + \kappa}\right)\\
&\subseteq B\left(0, 2\norm{b(\hat\gamma)}\xi^{-1} + \sqrt{\kappa}\right)\\
&\subseteq B\left(0,4\kappa+\sqrt{\kappa} \right)\\
&\subseteq B\left(0,5\kappa\right),
\end{align*}
where in the third line, we used Assumption~\ref{as:alg} and the bound $\norm{b(\hat\gamma)} \leq \norm{b_0} + \gamma_+\norm{b_1}\leq 2\zeta$. 
\end{proof}

\begin{lemma}
\label{lem:delta_uniform}
If $\norm{\hat x}\leq 5\kappa$, then $\tilde f(\hat x) \geq f(\hat x) - \epsilon/2$. In particular, $\widetilde\Opt \geq \Opt - \epsilon/2$,
\end{lemma}
\begin{proof}
Recall that $\delta \coloneqq \frac{\epsilon}{72\kappa^2}$.
Let $\hat x\in\R^n$ such that $\norm{\hat x}\leq 5\kappa$. We compute
\begin{align*}
\tilde f(\hat x) &= \max \set{q(\tilde\gamma_-, \hat x), q(\tilde\gamma_+, \hat x)}\\
&\geq \max \set{q(\gamma_-, \hat x), q(\gamma_+, \hat x)} - \delta \abs{q_1(\hat x)}\\
&\geq f(\hat x) - \delta \left(\norm{\hat x}^2 + 2\norm{\hat x} + 1\right)\\
&\geq f(\hat x) - \delta \left(6\kappa\right)^2\\
&= f(\hat x)-\epsilon/2,
\end{align*}
where the first inequality follows from \eqref{eq:gamma_ordering}, the second inequality follows from Assumption~\ref{as:alg}, and the third inequality follows from the bound $\norm{\hat x}\leq 5\kappa$.
\end{proof}

\subsection{Approximating $\gamma_-$ and $\gamma_+$}
\label{subsec:alg_tilde_gamma}

In this subsection, we show how to approximate $\gamma_-$ and $\gamma_+$ and provide an explicit running time analysis of this procedure. Our developments rely on the fact that $\lambda_{\min}(A(\gamma))$ is a concave function in $\gamma$ and that $\gamma_-$ and $\gamma_+$ are the unique zeros of this function.

\begin{lemma}
$\lambda_{\min}(A(\gamma))$ is a concave function in $\gamma$.
\end{lemma}
\begin{proof}
By Courant-Fischer Theorem,
$\lambda_{\min}(A(\gamma)) =\min_{\norm{x} = 1} x^\top A(\gamma) x$. Note that for any fixed $x\in\R^n$, the expression $x^\top A(\gamma) x$ is linear in $\gamma$. Then, the result follows upon recalling that the minimum of concave (in our case linear) functions is concave.
\end{proof}

Let us also state a simple property of the function $\lambda_{\min}(A(\gamma))$.
\begin{lemma}
\label{lemma:lambda_min_structure}
\leavevmode
\begin{enumerate}[(i)]
	\item Suppose $\gamma\leq\hat\gamma$, then 
	$
	\abs{\gamma-\gamma_-}\leq \kappa\abs{\lambda_{\min}(A(\gamma))}.
	$
	\item Suppose $\gamma\geq\hat\gamma$, then 
	$
	\abs{\gamma-\gamma_+}\leq \kappa\abs{\lambda_{\min}(A(\gamma))}.
	$
\end{enumerate}
\end{lemma}
\begin{proof}
We only prove the first statement as the second statement follows similarly. Let $\gamma\leq \hat\gamma$. From the concavity of $\lambda_{\min}(A(\gamma))$, we have
\begin{align*}
\abs{\lambda_{\min}(A(\gamma))}
\geq \abs{\gamma- \gamma_-} \frac{\lambda_{\min}(A(\hat\gamma))}{\hat\gamma - \gamma_-}
\geq \abs{\gamma- \gamma_-}\frac{\xi}{\zeta},
\end{align*}
where in the second inequality we used the definition of $\xi$ in Definition~\ref{def:regularity} and the bound $\hat\gamma - \gamma_- \leq \gamma_+ \leq \zeta$. Noting $\zeta/\xi = \kappa$ and rearranging terms completes the proof.
\end{proof}

We will use the Lanczos method for approximating the most negative eigenvalue (and a corresponding eigenvector) of a sparse matrix. This algorithm, along with Lemma \ref{lemma:lambda_min_structure}, will allow us to binary search over the range $[0,\zeta]$ for the zeros of the function $\lambda_{\min}(A(\gamma))$.

\begin{lemma}
[\cite{KuczynskiWozniakowski1992estimating}]
\label{lem:approx_eig}
There exists an algorithm, $\text{ApproxEig}(A,\rho,\eta,p_{\text{eig}})$, which given a symmetric matrix $A\in\R^{n\by n}$, $\rho$ such that $\norm{A}\leq \rho$, and parameters $\eta,p_{\text{eig}}>0$, will, with probability at least $1-p_{\text{eig}}$, return a unit vector $x\in\R^n$ such that $x^\top A x \leq \lambda_{\min}(A) +\eta$. This algorithm runs in time
\begin{align*}
O\left(\frac{N\sqrt{\rho}}{\sqrt{\eta}}\log\left(\frac{n}{p_{\text{eig}}}\right)\right),
\end{align*}
where $N$ is the number of nonzero entries in $A$.
\end{lemma}

Consider ApproxGammaPlus (Algorithm~\ref{alg:tilde_gamma}) for computing $\tilde\gamma_+$ up to accuracy $\delta$. A similar algorithm can be used to compute $\tilde\gamma_-$ up to accuracy $\delta$ and is omitted.
\begin{algorithm}
	\caption{$\text{ApproxGammaPlus}(q_0, q_1,\xi,\zeta,\hat\gamma,\delta,p_{\tilde\gamma_+})$}
	\label{alg:tilde_gamma}
	Given $q_0$ and $q_1$ satisfying Assumption \ref{as:alg}, $(\xi,\zeta)$ and $\hat\gamma$ satisfying Assumption \ref{as:alg_regularity}, error parameter $\delta>0$, and failure probability $p_{\tilde\gamma_+}$
	\begin{enumerate}
		\item Let $s_0 = \hat\gamma$, $t_0 = \zeta$
		\item Let $T=\ceil{\log\left(\frac{\zeta\kappa}{\delta}\right)}+2$
		\item For $k = 0,\dots, T-1$
		\begin{enumerate}
			\item Let $\gamma = (s_k + t_k)/2$
			\item Let $x = \text{ApproxEig}(A(\gamma), 2\zeta, \frac{\delta}{4\kappa}, \frac{p_{\tilde\gamma}}{T})$
			\item If $x^\top A(\gamma)x < \tfrac{\delta}{4\kappa}$, set $s_{k+1} = s_k$ and $t_{k+1} = \gamma$
			\item Else if $x^\top A(\gamma) x > \frac{\delta}{\kappa}$, set $s_{k+1} = \gamma$ and $t_{k+1} = t_k$
			\item Else, stop and output $\tilde\gamma$
		\end{enumerate}
	\end{enumerate}
\end{algorithm}

\begin{restatable}{lemma}{ApproxGammaPlusLemma}
\label{lem:tilde_gamma}
Given $q_0$, $q_1$ satisfying Assumption~\ref{as:alg}, $(\xi,\zeta)$ and $\hat\gamma$ satisfying Assumption~\ref{as:alg_regularity}, $\delta>0$, and $p_{\tilde\gamma_+}$, ApproxGammaPlus (Algorithm~\ref{alg:tilde_gamma}) outputs $\tilde\gamma_+$ satisfying
\begin{align*}
\tilde\gamma_+\in[\gamma_+-\delta,\;\gamma_+],\hspace{2em} \lambda_{\min}(A(\tilde\gamma_+))\leq \delta/\kappa
\end{align*}
with probability $1-p_{\tilde\gamma_+}$. This algorithm runs in time
\begin{align*}
\tilde O\left(\frac{N\sqrt{\kappa\zeta}}{\sqrt{\delta}}\log\left(\frac{n}{p_{\tilde\gamma_+}}\right)\log\left(\frac{\kappa}{\delta}\right)\right).
\end{align*}
\end{restatable}

\begin{proof}
We condition on the event that ApproxEig succeeds every time it is called. By the union bound, this happens with probability at least $1-p_{\tilde\gamma_+}$.

Suppose the algorithm outputs at step 3.(e). Let $\gamma$ be the value of $\gamma$ on the round in which the algorithm stops, and $x$ the vector returned by ApproxEig in the corresponding iteration. 
Then, the stopping rule guarantees $x^\top A(\gamma)x \in[\delta/4\kappa,\delta/\kappa]$. As we have conditioned on ApproxEig succeeding, we deduce
\begin{align*}
x^\top A(\gamma)x - \frac{\delta}{4\kappa}\leq \lambda_{\min}(A(\gamma))\leq x^\top A(\gamma)x.
\end{align*}
In particular, $\lambda_{\min}(A(\gamma))\in[0,\delta/\kappa]$ and $\gamma\leq\gamma_+$.
Applying Lemma \ref{lemma:lambda_min_structure} gives
\begin{align*}
\abs{\gamma - \gamma_+} &\leq \kappa\abs{\lambda_{\min}(A(\gamma))}\leq \delta.
\end{align*}
We conclude $\gamma_+-\delta\leq \gamma\leq \gamma_+$.

We now show that this algorithm outputs within $T$ rounds. Let
\begin{align*}
P\coloneqq \set{\gamma:\, \gamma\geq \hat\gamma,\, \lambda_{\min}(A(\gamma))\in[\delta/4\kappa,3\delta/4\kappa]}.
\end{align*}
Recalling that $\lambda_{\min}(A(\gamma))$ is $1$-Lipschitz in $\gamma$, we deduce that $\abs{P}\geq \delta/2\kappa$. Note also that $\lambda_{\min}(A(\hat\gamma))\geq \xi \geq \delta \geq 3\delta/4\kappa$ thus $P$ is a connected interval.

Suppose for the sake of contradiction that the algorithm fails to output in each of the $T$ rounds. Note that $P\subseteq[s_0,t_0]$.
We will show by induction that $P\subseteq[s_k,t_k]$ for every $k\in\set{1,\dots,T}$. Let $k\in\set{0,\dots,T-1}$. 
By assumption, the algorithm fails to output in round $k$. This can happen in two ways: If $x^\top A(\gamma)x < \delta/4\kappa$, then $x$ certifies that $\gamma\notin P$ and $P\subseteq[s_k, \gamma]$. If $x^\top A(\gamma)x>\delta/\kappa$, then as we have conditioned on ApproxEig succeeding, $\lambda_{\min}(A(\gamma))\geq \delta/\kappa - \delta/4\kappa$ and $P\subseteq [\gamma,t_k]$.
In either case, we have that $P\subseteq[s_{k+1},t_{k+1}]$.

We conclude that $P$, an interval of length at least $\delta/2\kappa$, is contained in $[s_T,t_T]$, an interval of length
\begin{align*}
t_T-s_T\leq \frac{\zeta}{2^T} \leq \delta/4\kappa,
\end{align*}
a contradiction. Thus, the algorithm outputs within $T$ rounds.

The running time of this algorithm follows from Lemma \ref{lem:approx_eig}.
\end{proof}

\begin{remark}
Similar algorithms for approximating $\gamma_\pm$ given $\hat\gamma$ have been proposed in the literature~\cite{more1993generalizations,PongWolkowicz2014,adachi2019eigenvalue,jiang2019novel}. However to our knowledge, this is the first analysis to establish an explicit convergence rate; see the discussion after Remark 2.11 in \cite{jiang2019novel} on this issue.
\end{remark}

\subsection{Minimizing the maximum of two quadratic functions}
\label{subsec:alg_nesterov}
In this subsection, we will assume that Algorithm~\ref{alg:GTRS_outer} has successfully found $\tilde\gamma_\pm$ satisfying \eqref{eq:tilde_gamma_guarantee} and show how to approximately solve
\begin{align*}
\min_{x\in\R^n}\max\set{q(\tilde\gamma_-, x),q(\tilde\gamma_+,x)}.
\end{align*}
For the sake of readability, we will use the following notation in this subsection.
\begin{align}
\label{eq:tilde_f_i}
\tilde f_0(x)\coloneqq q(\tilde\gamma_-, x)\quad \text{and} \quad
\tilde f_1(x)\coloneqq q(\tilde\gamma_+, x)
\end{align}
In particular we have $\tilde f(x) = \max\set{\tilde f_0(x), \tilde f_1(x)}$. 

Our analysis is based on \citet[Section 2.3.3]{nesterov2018lectures}, which proposes a high level algorithm for minimizing general minimax problems with smooth components. We state this algorithm (Algorithm \ref{alg:nesterov}) and its
corresponding convergence rate in our context.

\begin{algorithm}
\caption{Constant Step Scheme II for Smooth Minimax Problems \cite[Algorithm 2.3.12]{nesterov2018lectures}}
\label{alg:nesterov}
Given continuously differentiable convex, $2L$-smooth functions $\tilde f_0, \tilde f_1$
\begin{enumerate}
	\item Let $x_0 = y_0 = 0$ and $\alpha_0 = 1/2$
	\item For $k= 0,1,\dots$
	\begin{enumerate}
		\item Compute $\tilde f_i(y_k)$ and $\grad \tilde f_i(y_k)$ for $i=0,1$
		\item Compute
		\begin{align*}
		x_{k+1} &= \argmin_x \max_{i=0,1} \left(\tilde f_i(y_k) + \ip{\grad \tilde f_i(y_k), x-y_k} + L \norm{x- y_k}^2\right)\\
		\alpha_{k+1} &= \frac{\sqrt{\alpha_k^4 + 4\alpha_k^2} - \alpha_k^2}{2}\\
		\beta_k &= \frac{\alpha_k(1-\alpha_k)}{\alpha_k^2+\alpha_{k+1}}\\
		y_{k+1} &= x_{k+1} + \beta_k(x_{k+1} - x_k)
		\end{align*}
	\end{enumerate}
\end{enumerate}
\end{algorithm}

\begin{theorem}
[{\cite[Theorem 2.3.5]{nesterov2018lectures}}]
\label{thm:nesterov}
Let $\tilde f_0,\tilde f_1$ be $2L$-smooth\footnote{Recall that a convex quadratic function $x^\top A x+ 2b^\top x + c$ is $2L$-smooth if and only if $A\preceq LI$.} differentiable convex functions such that $\tilde f$ is bounded below. Let $\tilde x^*$ be an optimizer of $\tilde f$. Then the iterates $x_k$ produced by Algorithm \ref{alg:nesterov} satisfy
\begin{align*}
\tilde f(x_k) - \tilde f(\tilde x^*) \leq \frac{8}{(k+1)^2}\left(\tilde f(0) - \tilde f(\tilde x^*) + \frac{L}{2}\norm{\tilde x^*}^2\right).
\end{align*}
\end{theorem}

\begin{lemma}\label{lem:bound_q_i}
Let $x\in\R^n$. Then for $i=0,1$, we have
\[
\abs{q_i(0)-q_i(x)} \leq \norm{x}^2 + 2\norm{x}.
\]
\end{lemma}
\begin{proof}
For $i=0,1$, we have
\[
\abs{q_i(0)-q_i(x)}=\abs{q_i(x)-c_i} \leq \norm{A_i} \norm{x}^2  + 2\norm{b_i} \norm{x} \leq \norm{x}^2 + 2\norm{x}.
\]
where the second inequality follows from Assumption~\ref{as:alg}.
\end{proof}

\begin{corollary}
\label{cor:alg_iterations}
Let $\tilde f_0$ and $\tilde f_1$ be the functions defined in \eqref{eq:tilde_f_i}. Let $\tilde x^*$ be an optimizer of $\tilde f$.
Then the iterates $x_k$ produced by Algorithm \ref{alg:nesterov} satisfy
\begin{align*}
\tilde f(x_k) - \tilde f(\tilde x^*) \leq \frac{760}{(k+1)^2}\kappa^2\zeta.
\end{align*}
In particular, after $k = O\left(\kappa \sqrt{\zeta/\epsilon}\right)$ iterations, the solution $x_k$ satisfies $\tilde f(x_k)-\tilde f(\tilde x^*) \leq \epsilon/2$.
\end{corollary}
\begin{proof}
We have that $\tilde f_0$ and $\tilde f_1$ are both $2(2\zeta)$-smooth by Assumption \ref{as:alg} and Definition \ref{def:regularity}. Moreover, $\tilde f(x)\geq q(\hat\gamma,x)$ is bounded below. Thus, we may apply Theorem \ref{thm:nesterov}.

We bound the initial primal gap as follows:
\begin{align*}
\tilde f(0) - \tilde f(\tilde x^*) &= \max\set{\tilde f_0(0), \tilde f_1(0)} - \max\set{\tilde f_0(\tilde x^*),\tilde f_1(\tilde x^*)}\\
&\leq \max\set{\tilde f_0(0)-\tilde f_0(\tilde x^*),\tilde f_1(0)-\tilde f_1(\tilde x^*)}\\
&= q_0(0)-q_0(\tilde x^*) + \max\set{\tilde \gamma_- (q_1(0) - q_1(\tilde x^*)),\; \tilde \gamma_+ (q_1(0) - q_1(\tilde x^*))}\\
&\leq \abs{q_0(0) - q_0(\tilde x^*)} + \zeta \abs{q_1(0) - q_1(\tilde x^*)}\\
&\leq (1+\zeta)\left(25\kappa^2 + 10\kappa \right)\\
&\leq 70\kappa^2\zeta,
\end{align*}
where the third line follows from definition (see \eqref{eq:tilde_f_i}), the fourth line follows from the ordering $\tilde\gamma_-\leq\tilde\gamma_+\leq \zeta$, the fifth line follows from Lemmas~\ref{lem:locating_opt}~and~\ref{lem:bound_q_i}, and the last line follows from the trivial bounds $\kappa\geq 1$ and $\zeta\geq 1$. 

Using Lemma \ref{lem:locating_opt} again, we also have $\frac{L}{2}\norm{\tilde x^*}^2 = \frac{2\zeta}{2}\norm{\tilde x^*}^2 \leq 25\kappa^2\zeta$. The result follows by combining these bounds.
\end{proof}

It remains to analyze the runtime of each iteration. Aside from computation of $x_{k+1}$, it is clear that the quantities in each iteration can be computed in $O(N)$ time. Below, we derive a closed form expression for $x_{k+1}$ where each of the quantities can be computed in $O(N)$ time.

\begin{lemma}
\label{lem:alg_iteration_complexity}
For any  $y\in\R^n$, the quantity
\begin{align*}
\argmin_x \max_{i=0,1} \left(\tilde f_i(y) + \ip{\grad \tilde f_i(y), x-y} + L \norm{x- y}^2\right)
\end{align*}
can be computed in $O(N)$ time.
\end{lemma}
\begin{proof}
Fix $y\in\R^n$. We begin by recentering the quadratic functions in the objective.
\begin{align*}
&\max_{i=0,1}\left( \tilde f_i(y) + \ip{\grad \tilde f_i(y), x-y} + L \norm{x-y}^2\right)\\
&= \max_{i=0,1} \left(L \norm{x - \left[y - \frac{1}{L} \frac{\grad \tilde f_i(y)}{2}\right]}^2 + \left[\tilde f_i(y) - \frac{1}{L} \norm{\frac{\grad \tilde f_i(y)}{2}}^2\right]\right)\\
&\eqqcolon \max_{i=0,1} \left(L\norm{x-z_i}^2 + h_i\right)
\end{align*}
Here, $z_i$ and $h_i$ are defined to be the square-bracketed terms from the preceding line. It is clear that the minimizing $x$ must belong to the line segment $[z_0,z_1]$. We will parameterize $x = z_0 + \alpha(z_1-z_0)$ where $\alpha\in[0,1]$.
\begin{align*}
&\min_x \max_{i=0,1} \left(L\norm{x-z_i}^2 + h_i \right)\\
&= \min_{\alpha\in[0,1]} \max\left\{\alpha^2 L\norm{z_0-z_1}^2 + h_0,~(1-\alpha)^2L\norm{z_0-z_1}^2 + h_1\right\}.
\end{align*}
We solve for $\alpha$ by setting the two terms inside the maximum equal. A simple calculation yields that the two quadratics are equal when
\begin{align*}
\bar\alpha \coloneqq \frac{1}{2} - \frac{h_0 - h_1}{2L \norm{z_0-z_1}^2}.
\end{align*}
If $\bar\alpha$ is between $[0,1]$, let $\alpha^*=\bar\alpha$. Else let $\alpha^* = 0$ (respectively $\alpha^* = 1$) when $\bar\alpha<0$ (respectively $\bar\alpha>1$).

Then,
\begin{align*}
\argmin_x \max_{i=0,1} \left(\tilde f_i(y) + \ip{\grad \tilde f_i(y), x-y} + L \norm{x- y}^2\right)= z_0 + \alpha^* (z_1 - z_0).
\end{align*}
Each of the quantities on the right hand side (namely $\alpha^*$, $z_i$) can be computed in $O(N)$ time.
\end{proof}

Combining Corollary \ref{cor:alg_iterations} and Lemma \ref{lem:alg_iteration_complexity} gives the following corollary.
\begin{corollary}\label{cor:min_max_runtime}
Let $\tilde f_0, \tilde f_1$ be the functions defined in \eqref{eq:tilde_f_i}. There exists an algorithm which outputs $\tilde x$ satisfying $\tilde f(\tilde x) \leq \widetilde \Opt + \epsilon/2$ running in time
\begin{align*}
O\left(\frac{N\,\kappa\,\sqrt{\zeta}}{\sqrt{\epsilon}}\right).
\end{align*}
\end{corollary}

\begin{remark}\label{rem:runtime_novel_comparison}
\citet{jiang2019novel} present a saddle-point-based first-oder algorithm for approximating $\widetilde\Opt$. By instantiating their algorithm with the initial iterate $x_0 = 0$ and applying our Lemma~\ref{lem:locating_opt} to bound $\norm{x_0-\tilde x^*}^2$, we have that \cite[Algorithm 1]{jiang2019novel} produces an $\epsilon/2$-optimal solution to the convex reformulation in time
\begin{align*}
O\left(\frac{N\kappa^2\zeta}{\epsilon}\right).
\end{align*}
Therefore, the dependences on $\epsilon$, $\kappa$, and $\zeta$ of this algorithm are worse than that of the algorithm described in Corollary~\ref{cor:min_max_runtime}. 
Note that \cite{jiang2019novel} does not present an analysis of the complexity of finding the approximate generalized eigenvalue $\tilde\gamma_{\pm}$ (needed to construct $\widetilde\Opt$) or how $\widetilde\Opt$ relates to $\Opt$.
\end{remark}

By combining Lemmas~\ref{lem:ubOpt},~\ref{lem:delta_uniform},~and~\ref{lem:tilde_gamma}
and Corollary~\ref{cor:min_max_runtime}, we arrive at the following main theorem on the overall computational complexity of our approach.
\begin{theorem}
\label{thm:alg_runtime}
Given $q_0, q_1$ satisfying Assumption \ref{as:alg}, $(\xi,\zeta)$ and $\hat\gamma$ satisfying Assumption \ref{as:alg_regularity}, error parameter $0<\epsilon\leq \kappa^2\xi$, and failure probability $p>0$, ApproxConvex (Algorithm \ref{alg:GTRS_outer})  outputs $\tilde\gamma_-$, $\tilde\gamma_+$ and $\tilde x\in\R^n$ such that
\begin{align*}
\Opt \leq \max\set{q(\tilde\gamma_-, \tilde x), q(\tilde\gamma_+,\tilde x)} \leq \widetilde\Opt + \epsilon/2 \leq \Opt + \epsilon
\end{align*}
with probability $1-p$. This algorithm runs in time
\begin{align*}
\tilde O\left(\frac{N\kappa^{3/2}\sqrt{\zeta}}{\sqrt{\epsilon}}\log\left(\frac{n}{p}\right)\log\left(\frac{\kappa}{\epsilon}\right)\right).
\end{align*}
\end{theorem}

\subsection{Finding an approximate optimizer of the GTRS}\label{subsec:alg_rounding}

Let $\tilde x\in\R^n$ be the approximate optimizer output by Algorithm \ref{alg:GTRS_outer}.
In this subsection, we show how to use $\tilde x$ to construct an $\bar x$ approximately minimizing the original GTRS \eqref{eq:GTRS}. Our algorithm will follow the proof of Theorem~\ref{thm:pd_implies_conv_hull} (in particular Lemma~\ref{lem:pd:Ssupset}).

We present our algorithm, ApproxGTRS, as Algorithm~\ref{alg:rounding}. ApproxGTRS will use ApproxConvex as a subroutine. Given an additive error $\epsilon_{\text{round}}$, ApproxGTRS will call ApproxConvex with additive error $\epsilon_{\text{convex}}$. We will write these parameters as $\epsilon_r$ and $\epsilon_c$ for short.

\begin{algorithm}
	\caption{$\text{ApproxGTRS}(q_0,q_1,\xi,\zeta,\hat\gamma,\epsilon_r,p_r)$}
	\label{alg:rounding}
	Given $q_0$ and $q_1$ satisfying Assumption \ref{as:alg}, $(\xi,\zeta)$ and $\hat\gamma$ satisfying Assumption \ref{as:alg_regularity}, error parameter $0<\epsilon_r\leq \kappa^3\xi$, and failure probability $p_r>0$
	\begin{enumerate}
		\item Define $\epsilon_c \coloneqq \epsilon_r/(28\kappa)$
		\item Let $\tilde\gamma_-$, $\tilde\gamma_+$ and $\tilde x$ be the output of $\text{ApproxConvex}(q_0,q_1,\xi,\zeta,\hat\gamma, \epsilon_c,p_r / 2)$
		\item If $q_1(\tilde x) = 0$ then return $\bar x= \tilde x$
		\item Else if $q_1(\tilde x)>0$
		\begin{enumerate}
			\item Let	$d \coloneqq \text{ApproxEig}(A(\tilde \gamma_+), 2\zeta, \delta/\kappa, p_r/2)$
			\item Let $e \coloneqq 2\left(\tilde x^\top A(\tilde \gamma_+) d + b(\tilde \gamma_+)^\top d\right)$
			\item If necessary, take $d\gets -d$ and $e\gets -e$ to ensure that $e\leq 0$
			\item Let $\alpha\geq 0$ be the nonnegative solution to
			\begin{align*}
			q(\tilde\gamma_-, \tilde x+\alpha d) = q(\tilde \gamma_+, \tilde x+\alpha d)
			\end{align*}
			\item Return $\bar x = \tilde x + \alpha d$
		\end{enumerate}
		\item Else carry out the computation in step 4 where the roles of $\tilde \gamma_-$ and $\tilde\gamma_+$ are interchanged
	\end{enumerate}
\end{algorithm}
Note that by Definition \ref{def:regularity}, we have $\kappa^3\xi\geq 1$. Thus, as before, the requirement $0<\epsilon_r\leq \kappa^3\xi$ in Algorithm \ref{alg:rounding} is not a practical issue: given $\epsilon_r>\kappa^3\xi$, we can simply run our algorithm with $\epsilon_r' = \kappa^3/\xi$ and return a solution with a better error guarantee.

The next lemma bounds $\norm{\tilde x}$. Its proof follows the proof of Lemma \ref{lem:locating_opt} with minor adjustments (in particular, the upper bound of \eqref{eq:locating_opt_upper_bound} is replaced with $\tilde f(\tilde x)\leq \max\set{c(\gamma_-),c(\gamma_+)} + \epsilon/2$; see Corollary~\ref{cor:min_max_runtime}) and is omitted.
\begin{lemma}\label{lem:locating_tilde_x}
Let $\tilde x\in\R^n$ satisfy $\tilde f(\tilde x)\leq \tilde \Opt + \epsilon/2$. Then 
$\tilde x \in B(0, 6\kappa)$.
\end{lemma}

We are now ready to prove a formal guarantee on Algorithm \ref{alg:rounding}.
\begin{theorem}
Given $q_0,q_1$ satisfying Assumption \ref{as:alg}, $(\xi,\zeta)$ and $\hat\gamma$ satisfying Assumption \ref{as:alg_regularity}, error parameter $0<\epsilon_r\leq \kappa^3\xi$, and failure probability $p_r$, ApproxGTRS (Algorithm \ref{alg:rounding}) outputs $\bar x$ such that
\begin{align*}
q_0(\bar x) &\leq \Opt + \epsilon_r \\
q_1(\bar x) &=0
\end{align*}
with probability $1-p_r$. 
This algorithm runs in time
\begin{align*}
\tilde O\left(\frac{N\kappa^{2}\sqrt{\zeta}}{\sqrt{\epsilon_r}}\log\left(\frac{n}{p_r}\right)\log\left(\frac{\kappa}{\epsilon_r}\right)\right).
\end{align*}
\end{theorem}
\begin{proof}
We condition on the event that Algorithm \ref{alg:GTRS_outer} succeeds and the ApproxEig call in step 4.(a) or 5.(a) succeeds. By the union bound, this happens with probability at least $1-p_r$. As in Lemma~\ref{lem:pd:Ssupset}, we will split the analysis into three cases:
(i) $q_1(\tilde x) = 0$,
(ii) $q_1(\tilde x)>0$, and
(iii) $q_1(\tilde x)<0$.

\begin{enumerate}[(i)]
	\item If $q_1(\tilde x) = 0$ then $q_0(\tilde x) = \tilde f(\tilde x) \leq \widetilde \Opt + \epsilon_c/2 \leq \Opt + \epsilon_c\leq \Opt + \epsilon_r$.
	\item Now suppose $q_1(\tilde x)>0$, i.e., we are in step 4 of Algorithm \ref{alg:rounding}.
We will need an upper bound on the value of $\alpha$ found in step 4.(d).

Let $t \coloneqq q(\tilde\gamma_+, \tilde x)$.
Recall that $\lambda_{\min}(A(\tilde\gamma_+))\in[0,\delta/\kappa]$ (see Lemma \ref{lem:tilde_gamma}).
Then, as we have conditioned on the ApproxEig call in step 4.(a) succeeding, we have
\begin{equation}
\label{eq:q_bar_x_upper_bound}
\begin{aligned}
q(\tilde\gamma_+,\tilde x + \alpha d) - (t+\alpha e) &= (q(\tilde\gamma_+, \tilde x) - t) + \alpha^2 d^\top A(\tilde\gamma_+)d\\
&\leq \alpha^2 (2\delta/\kappa).
\end{aligned}
\end{equation}
Next, we give a lower bound on $q(\tilde\gamma_-,\tilde x + \alpha d) - (t+\alpha e)$ using the estimate $d^\top A(\tilde\gamma_-)d \geq \xi$, and routine estimates on $\norm{A(\gamma)}$ and $\norm{b(\gamma)}$:
\begin{align*}
&q(\tilde\gamma_-,\tilde x + \alpha d) - (t+\alpha e)\\
&\qquad= (q(\tilde\gamma_-, \tilde x) - t) + 2\alpha \left(\tilde x^\top A (\tilde\gamma_-) d + b(\tilde\gamma_-)^\top d - e/2\right) + \alpha^2 d^\top A(\tilde\gamma_-)d\\
&\qquad\geq -\abs{\tilde\gamma_+-\tilde\gamma_-}\abs{q_1(\tilde x)}\\
&\qquad\hphantom{\geq} - 2\alpha\left(\norm{\tilde x} \norm{A(\tilde\gamma_-)} + \norm{b(\tilde\gamma_-)}+\norm{\tilde x} \norm{A(\tilde\gamma_+)} + \norm{b(\tilde\gamma_+)}\right)\\
&\qquad\hphantom{\geq} +\alpha^2 \xi\\
&\qquad\geq -49\kappa^2\zeta - 2\alpha \left(14\kappa\zeta\right) + \alpha^2 \xi,
\end{align*}
where the last inequality follows from the bounds $\abs{\tilde\gamma_+-\tilde\gamma_-}\leq\gamma_+\leq \zeta$ (Definition \ref{def:regularity}), $\norm{\tilde x}\leq 6\kappa$ (Lemma \ref{lem:locating_tilde_x}), and Lemma~\ref{lem:bound_q_i}.

We may combine our upper and lower bounds to deduce that for any $\alpha\in\R$,
\begin{align*}
q(\tilde\gamma_-, \tilde x + \alpha d) -
q(\tilde\gamma_+,\tilde x + \alpha d)  &\geq \alpha^2(\xi-2\delta/\kappa)-2\alpha(14\kappa\zeta)-49\kappa^2\zeta\\
&\geq \alpha^2\left(\frac{35}{36}\xi\right) -2\alpha(14\kappa\zeta)-49\kappa^2\zeta,
\end{align*}
where the last relation follows from the definition of $\delta$ in \eqref{eq:delta}, the definition of $\epsilon_c$ and the assumption on $\epsilon_r$, we have $\epsilon_c\leq \kappa^2\xi$, and the bound $\kappa\geq 1$.
In particular, because the quadratic function on the left is negative at $\alpha = 0$ and is lower bounded by a strongly convex quadratic function, there must exist both positive and negative choices of $\alpha$ for which the left hand side takes the value zero. This justifies step 4.(d) of the algorithm.

We now fix $\alpha$ to be the positive solution to $q(\tilde \gamma_-,\tilde x +\alpha d) = q(\tilde \gamma_+,\tilde x + \alpha d)$ so that
\begin{align*}
0 \geq \alpha^2\left(\frac{35}{36}\xi\right) -2\alpha(14\kappa\zeta)-49\kappa^2\zeta.
\end{align*}

We get an upper bound on $\alpha$ by the quadratic formula
\begin{align*}
\alpha &\leq \frac{14\kappa\zeta + \sqrt{(14\kappa\zeta)^2 + \tfrac{49\cdot 35}{36} \kappa^2\zeta\xi}}{\tfrac{35}{36}\xi}
\leq 31 \kappa^2.
\end{align*}
Then, by defining $\bar x \coloneqq \tilde x + \alpha d$, we have $q(\tilde\gamma_-,\bar x)=q(\tilde\gamma_+,\bar x)$.
Note that the containment $\hat\gamma\in(\tilde\gamma_-,\tilde\gamma_+)$ from \eqref{eq:gamma_ordering} implies $\tilde\gamma_-\neq\tilde\gamma_+$. Then we deduce $q_1(\bar x) =0$. Moreover, our upper bound \eqref{eq:q_bar_x_upper_bound} gives
\begin{align*}
q_0(\bar x) &= q(\tilde\gamma_+,\bar x)\\
&\leq t + \alpha e + \alpha^2(2\delta/\kappa).
\end{align*}
Then recalling that $t\coloneqq q(\tilde\gamma_+, \tilde x) \leq \widetilde \Opt + \epsilon_c/2 \leq \Opt + \epsilon_c$ and that we picked $e\leq 0$, we bound
\begin{align*}
q_0(\bar x) &\leq \Opt + \epsilon_c + (31\kappa^2)^2 (2\delta/\kappa)\\
&\leq \Opt +\epsilon_c + 27\kappa\epsilon_c\\
&\leq \Opt +28\kappa\epsilon_c\\
& = \Opt + \epsilon_r.
\end{align*}
\item The final case is symmetric to case (ii) and is omitted.
\end{enumerate}

The running time of this algorithm follows from Lemma \ref{lem:approx_eig} and Theorem \ref{thm:alg_runtime}.
\end{proof}

\begin{remark}\label{rem:runtime_comparison}
Let us now compare the running time of our algorithms to the running time of the algorithm presented by \citet{jiang2020linear}. 
This algorithm takes as input a pair $q_0$, $q_1$ satisfying an assumption similar to our Assumption~\ref{as:alg} and a regularity parameter $\xi_{\text{JL}}$. See Remark~\ref{rem:regularity_comparison} for a discussion of how the parameter $\xi_{\text{JL}}$ relates to our regularity parameters $(\xi_{\text{us}},\zeta)$.
Then given $\epsilon>0$ and $p>0$, this algorithm returns an $\epsilon$-optimal feasible solution with probability at least $1-p$. The running time of this algorithm is
\begin{align*}
\tilde O\left(\frac{N\phi^3}{\sqrt{\epsilon\, \xi_{\text{JL}}^5}}\log\left(\frac{n}{p}\right)\log\left(\frac{\phi}{\epsilon\,\xi_{\text{JL}}}\right)\right),
\end{align*}
where $\phi$ is a computable regularity parameter.

Recall that in Remark~\ref{rem:regularity_comparison},
we constructed simple examples where $\xi_{\text{JL}}\approx 1/2\kappa$ and $\zeta \approx 1$.
One can check that the regularity parameter $\phi$ is a constant on these examples.
In particular, the analysis presented in \citet{jiang2020linear} implies a running time of
\begin{align*}
\tilde O\left(\frac{N\kappa^{5/2}}{\sqrt{\epsilon}}\log\left(\frac{n}{p}\right)\log\left(\frac{\kappa}{\epsilon}\right)\right)
\end{align*}
on these instances. We contrast this with the running times
\begin{align*}
\tilde O\left(\frac{N\kappa^{3/2}}{\sqrt{\epsilon}}\log\left(\frac{n}{p}\right)\log\left(\frac{\kappa}{\epsilon}\right)\right)
,\qquad
\tilde O\left(\frac{N\kappa^2}{\sqrt{\epsilon}}\log\left(\frac{n}{p}\right)\log\left(\frac{\kappa}{\epsilon}\right)\right)
\end{align*}
of our Algorithms~\ref{alg:GTRS_outer}~and~\ref{alg:rounding} for finding an $\epsilon$-optimal value, and an $\epsilon$-optimal feasible solution respectively on these instances.
\end{remark}

\begin{remark}
Algorithms~\ref{alg:GTRS_outer}~and~\ref{alg:rounding} were designed and analyzed with worst-case guarantees in mind.
Consequently, we have not been particularly careful about bounding the constants in our analysis (for example the bounds $\kappa,\zeta,\xi^{-1}\leq 1$ are routinely used).
As such, there may be variants of our algorithms that achieve the \textit{same} worst-case guarantees with significantly faster numerical performance.
Similarly, the algorithm presented by \citet{jiang2020linear} is analyzed with worst-case guarantees in mind. They also remark that the the numerical performance of their algorithm may improve ``with suitable modifications'' (see \citet[Remark 4.2]{jiang2020linear}).

We leave such implementation questions and a thorough comparison of the numerical performance of the algorithms present in the literature for future work.
\end{remark}

\subsection{Further remarks}
\label{subsec:alg_further_remarks}
The algorithms given in the prior subsections can be sped up substantially if we know how to compute $\gamma_\pm$ and the corresponding zero eigenvectors exactly. As an example, we consider the special case where $A_0$ and $A_1$ are diagonal matrices.

\begin{lemma}
\label{lem:computing_diag_regularity}
There exists an algorithm which given $q_0$, $q_1$ satisfying Assumption~\ref{as:alg} with $A_0$ and $A_1$ diagonal, returns $\gamma_\pm$, $(\xi^*,\zeta^*)$ and $\gamma^*$ such that $\lambda_{\min}(A(\gamma^*)) = \xi^*$ in time $O(n)$.
\end{lemma}
\begin{proof}
Let $a_0$, $a_1\in\R^n$ be the diagonal entries of $A_0$ and $A_1$ respectively.
Note that
\begin{align*}
\lambda_{\min}(A(\gamma)) &= \min_{i\in[n]} \set{a_{0,i} + \gamma a_{1,i}}.
\end{align*}
Thus, $\gamma_\pm$ and $\zeta^*$ can clearly be computed in $O(n)$ time. Note that  
\begin{align*}
\xi^* &= \max_{\gamma,\xi}\set{\xi :\, \begin{array}{l}
	\forall i\in[n],\, a_{0,i} + \gamma a_{1,i} \geq \xi\\
	\xi\geq 0
\end{array}}.
\end{align*}
Hence, $\xi^*$ and $\gamma^*$ are, respectively, the optimal value and solution to a two-variable linear program with $n$ constraints.
Applying the algorithm by \citet{megiddo1983linear} for two-variable linear programming allows us to solve for $\xi^*$ and $\gamma^*$ in $O(n)$ time.
\end{proof}

\begin{corollary}
\label{cor:diagonal_opt}
There exists an algorithm which given $q_0$, $q_1$ satisfying Assumption \ref{as:alg} with $A_0$ and $A_1$ diagonal and error parameter $\epsilon>0$, outputs $\bar x\in\R^n$ such that
\begin{align*}
q_0(\bar x) &\leq \Opt + \epsilon\\
q_1(\bar x) &= 0.
\end{align*}
This algorithm runs in time
\begin{align*}
O\left(\frac{n\kappa^*\sqrt{\zeta^*}}{\sqrt{\epsilon}}\right).
\end{align*}
\end{corollary}
\begin{proof}
When $A_0$ and $A_1$ are diagonal we have $N \leq 2n$.
By Lemma~\ref{lem:computing_diag_regularity}, we can compute all of the quantities needed for the exact convex reformulation in $O(n)$ time.
Algorithm \ref{alg:nesterov} can then be applied to the exact convex reformulation to find $\tilde x\in\R^n$ with
\begin{align*}
\max\set{q(\gamma_-, \tilde x),q(\gamma_+,\tilde x)} \leq \Opt + \epsilon.
\end{align*}
We can further carry out the modification procedure of Lemma~\ref{lem:pd:Ssupset} exactly in $O(n)$ time.

The running time of this algorithm follows from Corollary~\ref{cor:min_max_runtime}.
\end{proof}

\section*{Acknowledgments}
This research is supported in part by NSF grant CMMI 1454548.

{
\bibliographystyle{plainnat}

\begin{thebibliography}{37}
\providecommand{\natexlab}[1]{#1}
\providecommand{\url}[1]{\texttt{#1}}
\expandafter\ifx\csname urlstyle\endcsname\relax
  \providecommand{\doi}[1]{doi: #1}\else
  \providecommand{\doi}{doi: \begingroup \urlstyle{rm}\Url}\fi

\bibitem[Adachi and Nakatsukasa(2019)]{adachi2019eigenvalue}
Satoru Adachi and Yuji Nakatsukasa.
\newblock Eigenvalue-based algorithm and analysis for nonconvex {QCQP} with one
  constraint.
\newblock \emph{Mathematical Programming}, 173\penalty0 (1):\penalty0 79--116,
  2019.

\bibitem[Ben-Tal and den Hertog(2014)]{BenTalDenHertog2014}
Aharon Ben-Tal and Dick den Hertog.
\newblock Hidden conic quadratic representation of some nonconvex quadratic
  optimization problems.
\newblock \emph{Mathematical Programming}, 143\penalty0 (1):\penalty0 1--29,
  2014.

\bibitem[Ben-Tal and Teboulle(1996)]{BenTalTeboulle1996}
Aharon Ben-Tal and Marc Teboulle.
\newblock Hidden convexity in some nonconvex quadratically constrained
  quadratic programming.
\newblock \emph{Mathematical Programming}, 72\penalty0 (1):\penalty0 51–63,
  1996.

\bibitem[Buchheim et~al.(2013)Buchheim, De~Santis, Palagi, and
  Piacentini]{BuchheimDSPP2013}
Christoph Buchheim, Marianna De~Santis, Laura Palagi, and Mauro Piacentini.
\newblock An exact algorithm for nonconvex quadratic integer minimization using
  ellipsoidal relaxations.
\newblock \emph{SIAM Journal on Optimization}, 23\penalty0 (3):\penalty0
  1867--1889, 2013.

\bibitem[Burer and K{\i}l{\i}n\c{c}-Karzan(2017)]{BKK14}
Samuel Burer and Fatma K{\i}l{\i}n\c{c}-Karzan.
\newblock How to convexify the intersection of a second order cone and a
  nonconvex quadratic.
\newblock \emph{Mathematical Programming}, 162\penalty0 (1):\penalty0 393--429,
  2017.

\bibitem[Conn et~al.(2000)Conn, Gould, and Toint]{Conn.et.al.2000}
Andrew~R. Conn, Nicholas I.~M. Gould, and Phillippe~L. Toint.
\newblock \emph{Trust Region Methods}.
\newblock MOS-SIAM Series on Optimization. SIAM, Philadelphia, PA, USA, 2000.

\bibitem[Fallahi et~al.(2018)Fallahi, Salahi, and
  Terlaky]{fallahi2018minimizing}
S.~Fallahi, M.~Salahi, and T.~Terlaky.
\newblock Minimizing an indefinite quadratic function subject to a single
  indefinite quadratic constraint.
\newblock \emph{Optimization}, 67\penalty0 (1):\penalty0 55--65, 2018.

\bibitem[Feng et~al.(2012)Feng, Lin, Sheu, and Xia]{feng2012duality}
Joe-Mei Feng, Gang-Xuan Lin, Reuy-Lin Sheu, and Yong Xia.
\newblock Duality and solutions for quadratic programming over single
  non-homogeneous quadratic constraint.
\newblock \emph{Journal of Global Optimization}, 54\penalty0 (2):\penalty0
  275--293, 2012.

\bibitem[Fortin and Wolkowicz(2004)]{FortinWolkowicz2004}
Charles Fortin and Henry Wolkowicz.
\newblock The {Trust Region Subproblem} and semidefinite programming.
\newblock \emph{Optimization Methods and Software}, 19\penalty0 (1):\penalty0
  41--67, 2004.

\bibitem[Fradkov and Yakubovich(1979)]{FradkovYakubovich1979}
Alexander~L. Fradkov and Vladimir~A. Yakubovich.
\newblock The {S}-procedure and duality relations in nonconvex problems of
  quadratic programming.
\newblock \emph{Vestn. LGU, Ser. Mat., Mekh., Astron}, 6\penalty0 (1):\penalty0
  101--109, 1979.

\bibitem[Gould et~al.(1999)Gould, Lucidi, Roma, and Toint]{Gould_LRT_99}
Nicholas I.~M. Gould, Stefano Lucidi, Massimo Roma, and Philippe~L. Toint.
\newblock Solving the {Trust-Region Subproblem} using the {L}anczos method.
\newblock \emph{SIAM Journal on Optimization}, 9\penalty0 (2):\penalty0
  504--525, 1999.

\bibitem[Guo et~al.(2009)Guo, Higham, and Tisseur]{guo2009improved}
Chun-Hua Guo, Nicholas~J Higham, and Fran{\c{c}}oise Tisseur.
\newblock An improved arc algorithm for detecting definite hermitian pairs.
\newblock \emph{SIAM Journal on Matrix Analysis and Applications}, 31\penalty0
  (3):\penalty0 1131--1151, 2009.

\bibitem[Hazan and Koren(2016)]{HazanKoren2016}
Elad Hazan and Tomer Koren.
\newblock A linear-time algorithm for trust region problems.
\newblock \emph{Mathematical Programming}, 158\penalty0 (1):\penalty0 363--381,
  2016.

\bibitem[Ho-Nguyen and K{\i}l{\i}n{\c{c}}-Karzan(2018)]{Ho-NguyenKK2016RO}
Nam Ho-Nguyen and Fatma K{\i}l{\i}n{\c{c}}-Karzan.
\newblock Online first-order framework for robust convex optimization.
\newblock \emph{Operations Research}, 66\penalty0 (6):\penalty0 1670--1692,
  2018.

\bibitem[Ho-Nguyen and K{\i}l{\i}n\c{c}-Karzan(2017)]{Ho-NguyenKK2017}
Nam Ho-Nguyen and Fatma K{\i}l{\i}n\c{c}-Karzan.
\newblock A second-order cone based approach for solving the {Trust Region
  Subproblem} and its variants.
\newblock \emph{SIAM Journal on Optimization}, 27\penalty0 (3):\penalty0
  1485--1512, 2017.

\bibitem[Jiang and Li(2016)]{jiang2016simultaneous}
Rujun Jiang and Duan Li.
\newblock Simultaneous diagonalization of matrices and its applications in
  quadratically constrained quadratic programming.
\newblock \emph{SIAM Journal on Optimization}, 26\penalty0 (3):\penalty0
  1649--1668, 2016.

\bibitem[Jiang and Li(2019)]{jiang2019novel}
Rujun Jiang and Duan Li.
\newblock Novel reformulations and efficient algorithms for the {Generalized
  Trust Region Subproblem}.
\newblock \emph{SIAM Journal on Optimization}, 29\penalty0 (2):\penalty0
  1603--1633, 2019.

\bibitem[Jiang and Li(2020)]{jiang2020linear}
Rujun Jiang and Duan Li.
\newblock A linear-time algorithm for generalized trust region subproblems.
\newblock \emph{SIAM Journal on Optimization}, 30\penalty0 (1):\penalty0
  915--932, 2020.

\bibitem[Jiang et~al.(2018)Jiang, Li, and Wu]{jiang2018socp}
Rujun Jiang, Duan Li, and Baiyi Wu.
\newblock {SOCP} reformulation for the {Generalized Trust Region Subproblem}
  via a canonical form of two symmetric matrices.
\newblock \emph{Mathematical Programming}, 169\penalty0 (2):\penalty0 531--563,
  2018.

\bibitem[K{\i}l{\i}n\c{c}-Karzan and Y{\i}ld{\i}z(2015)]{KKY15}
Fatma K{\i}l{\i}n\c{c}-Karzan and Sercan Y{\i}ld{\i}z.
\newblock Two-term disjunctions on the second-order cone.
\newblock \emph{Mathematical Programming}, 154\penalty0 (1):\penalty0 463--491,
  2015.

\bibitem[Kuczynski and Wozniakowski(1992)]{KuczynskiWozniakowski1992estimating}
J.~Kuczynski and H.~Wozniakowski.
\newblock Estimating the largest eigenvalue by the power and lanczos algorithms
  with a random start.
\newblock \emph{SIAM Journal on Matrix Analysis and Applications}, 13\penalty0
  (4):\penalty0 1094--1122, 1992.

\bibitem[Locatelli(2015)]{locatelli2015some}
Marco Locatelli.
\newblock Some results for quadratic problems with one or two quadratic
  constraints.
\newblock \emph{Operations Research Letters}, 43\penalty0 (2):\penalty0
  126--131, 2015.

\bibitem[Megiddo(1983)]{megiddo1983linear}
Nimrod Megiddo.
\newblock Linear-time algorithms for linear programming in {$\R^3$} and related
  problems.
\newblock \emph{SIAM Journal on Computing}, 12\penalty0 (4):\penalty0 759--776,
  1983.

\bibitem[Modaresi and Vielma(2017)]{modaresi2017convex}
Sina Modaresi and Juan~Pablo Vielma.
\newblock Convex hull of two quadratic or a conic quadratic and a quadratic
  inequality.
\newblock \emph{Mathematical Programming}, 164\penalty0 (1-2):\penalty0
  383--409, 2017.

\bibitem[Mor{\'e}(1993)]{more1993generalizations}
Jorge~J. Mor{\'e}.
\newblock Generalizations of the {Trust Region Problem}.
\newblock \emph{Optimization methods and Software}, 2\penalty0 (3-4):\penalty0
  189--209, 1993.

\bibitem[Mor{\'e} and Sorensen(1983)]{More_Sorensen_83}
Jorge~J. Mor{\'e} and Danny~C. Sorensen.
\newblock Computing a trust region step.
\newblock \emph{SIAM Journal on Scientific and Statistical Computing},
  4\penalty0 (3):\penalty0 553--572, 1983.

\bibitem[Nesterov(2018)]{nesterov2018lectures}
Yurii Nesterov.
\newblock \emph{Lectures on convex optimization (2nd Ed.)}.
\newblock Springer Optimization and Its Applications. Springer International
  Publishing, Basel, Switzerland, 2018.

\bibitem[Nocedal and Wright(2006)]{NocedalWright2000numerical}
Jorge Nocedal and Stephen~J. Wright.
\newblock \emph{Numerical Optimization}.
\newblock Springer Series in Operations Research and Financial Engineering.
  Springer, New York, NY, USA, 2006.

\bibitem[P{\'o}lik and Terlaky(2007)]{PolikTerlaky2007}
Imre P{\'o}lik and Tam{\'a}s Terlaky.
\newblock A survey of the {S}-lemma.
\newblock \emph{SIAM Review}, 49\penalty0 (3):\penalty0 371--418, 2007.

\bibitem[Pong and Wolkowicz(2014)]{PongWolkowicz2014}
Ting~Kei Pong and Henry Wolkowicz.
\newblock The {Generalized Trust Region Subproblem}.
\newblock \emph{Computational Optimization and Applications}, 58\penalty0
  (2):\penalty0 273--322, 2014.

\bibitem[Rendl and Wolkowicz(1997)]{Rendl_Wolkowicz_97}
Franz Rendl and Henry Wolkowicz.
\newblock A semidefinite framework for trust region subproblems with
  applications to large scale minimization.
\newblock \emph{Mathematical Programming}, 77\penalty0 (2):\penalty0 273--299,
  1997.

\bibitem[Salahi and Taati(2018)]{salahi2018efficient}
Maziar Salahi and Akram Taati.
\newblock An efficient algorithm for solving the {Generalized Trust Region
  Subproblem}.
\newblock \emph{Computational and Applied Mathematics}, 37\penalty0
  (1):\penalty0 395--413, 2018.

\bibitem[Stern and Wolkowicz(1995)]{SternWolkowicz1995}
Ronald~J. Stern and Henry Wolkowicz.
\newblock Indefinite trust region subproblems and nonsymmetric eigenvalue
  perturbations.
\newblock \emph{SIAM Journal on Optimization}, 5\penalty0 (2):\penalty0
  286--313, 1995.

\bibitem[Wang and K{\i}l{\i}n\c{c}-Karzan(2019)]{wang2019tightness}
Alex~L. Wang and Fatma K{\i}l{\i}n\c{c}-Karzan.
\newblock {On the tightness of SDP relaxations of QCQPs}.
\newblock Technical Report arXiv:1911.09195, ArXiV, 2019.
\newblock URL \url{https://arxiv.org/abs/1911.09195}.

\bibitem[Wang and Xia(2017)]{WangXia2016}
Jiulin Wang and Yong Xia.
\newblock A linear-time algorithm for the {Trust Region Subproblem} based on
  hidden convexity.
\newblock \emph{Optimization Letters}, 11\penalty0 (8):\penalty0 1639--1646,
  2017.

\bibitem[Yang et~al.(2018)Yang, Anstreicher, and Burer]{yang2018quadratic}
Boshi Yang, Kurt Anstreicher, and Samuel Burer.
\newblock Quadratic programs with hollows.
\newblock \emph{Mathematical Programming}, 170\penalty0 (2):\penalty0 541--553,
  2018.

\bibitem[Y{\i}ld{\i}ran(2009)]{yildiran2009convex}
U{\u{g}}ur Y{\i}ld{\i}ran.
\newblock Convex hull of two quadratic constraints is an {LMI} set.
\newblock \emph{IMA Journal of Mathematical Control and Information},
  26\penalty0 (4):\penalty0 417--450, 2009.

\end{thebibliography}

}

\appendix
\section{Proofs of Theorems \ref{thm:pd_implies_conv_hull_general} and \ref{thm:psd_implies_conv_hull_general}}
\label{app:general_conv_hull_proofs}

In this appendix, we outline how to modify the proofs of Theorems~\ref{thm:pd_implies_conv_hull}~and~\ref{thm:psd_implies_conv_hull} to prove Theorems~\ref{thm:pd_implies_conv_hull_general}~and~\ref{thm:psd_implies_conv_hull_general}.

\pdGeneralConvHull*
\begin{proof}
The ``$\subseteq$'' inclusions follow from a trivial modification of Lemma~\ref{lem:pd:Ssubset}. It suffices to prove the ``$\supseteq$'' inclusions.
The case where $A_0$ and $A_1$ are both nonconvex is covered by Theorem~\ref{thm:pd_implies_conv_hull}. We consider the four remaining cases:
\begin{itemize}
	\item Suppose $A_0$ and $A_1$ are both convex.
	In this case, $\Gamma = [0,\infty)$ and it suffices to show that $\conv(\cS) = \set{(x,t):\, q_0(x)\leq t,\, q_1(x) \leq 0} = \cS$. This holds as $\cS$ is convex.

	\item Suppose $A_0$ is nonconvex and $A_1$ is convex. In this case, $\Gamma = [\gamma_-, \infty)$ is unbounded above. Furthermore, $\gamma_-$ is positive and $A(\gamma_-)$ has a zero eigenvalue. Suppose $(\hat x,\hat t)$ satisfies $q(\gamma_-, \hat x)\leq \hat t$ and $q_1(\hat x)\leq 0$. If $q_1(\hat x) = 0$, then we also have $q_0(\hat x) = q(\gamma_-,\hat x) \leq \hat t$, whence $(\hat x,\hat t)\in\cS$. On the other hand, if $q_1(\hat x) < 0$, we may apply the argument in case (iii) in the proof of Lemma~\ref{lem:pd:Ssupset} verbatim (after replacing all occurrences of $\gamma_+$ by $\gamma^*$) to conclude that $(\hat x,\hat t)\in\conv(\cS)$.

	\item Suppose $A_0$ is convex and $A_1$ is nonconvex. In this case, $\Gamma = [0,\gamma_+]$ is bounded above and $\gamma_-$ is defined to be $\gamma_-=0$. Furthermore, $A(\gamma_+)$ has a zero eigenvalue. Suppose $(\hat x,\hat t)\in\cS(\gamma_-)\cap\cS(\gamma_+)$. If $q_1(\hat x) \leq 0$, then we also have $q_0(\hat x) = q(\gamma_-,\hat x) \leq \hat t$, whence $(\hat x,\hat t)\in\cS$. On the other hand, if $q_1(\hat x) > 0$, we may apply the argument in case (ii) in the proof of Lemma~\ref{lem:pd:Ssupset} verbatim to conclude that $(\hat x,\hat t)\in\conv(\cS)$.\qedhere
\end{itemize}
\end{proof}

We will prove Theorem~\ref{thm:psd_implies_conv_hull_general} using a limiting argument and reducing it to Theorem~\ref{thm:pd_implies_conv_hull_general}. The proof follows that of Lemma~\ref{lem:psd:Ssupset} almost verbatim.

\psdGeneralConvHull*
\begin{proof}
The ``$\subseteq$'' inclusions follow from a trivial modification of Lemma~\ref{lem:psd:Ssubset}. It suffices to prove the ``$\supseteq$'' inclusions.

Denote the set on the right hand side by $\cR$, i.e., $\cR \coloneqq \cS(\gamma_-) \cap\cS(\gamma_+)$ when $\Gamma$ is bounded and $\cR\coloneqq \cS(\gamma_-)\cap\set{(x,t):\, q_1(x)\leq 0}$ when $\Gamma$ is unbounded.

Let $(\hat x,\hat t)\in\cR$. It suffices to show that $(\hat x,\hat t+\epsilon)\in\conv(\cS)$ for all $\epsilon>0$.

We will perturb $A_0$ slightly to create a new instance of the problem. Let $\delta>0$ to be picked later. Define $A_0' = A_0 +\delta I_n$ and let all remaining data be unchanged, i.e.,
\begin{align*}
q_0'(x) &\coloneqq x^\top A_0' x + 2b_0'^\top x + c_0' \coloneqq x^\top (A_0+\delta I_n) x + 2b_0^\top x + c_0\\
q_1'(x) &\coloneqq x^\top A_1' x + 2b_1'^\top x + c_1' \coloneqq x^\top A_1 x + 2b_1^\top x + c_1.
\end{align*}
We will denote all quantities related to the perturbed system with an apostrophe.

We claim it suffices to show that there exists $\delta>0$ small enough such that $(\hat x,\hat t+\epsilon)\in\cR'$. Indeed, suppose this is the case. Note that for any $x\in\R^n$, we have $q_1(x) = q_1'(x)$ and $q_0(x) \leq q_0'(x)$. Hence, $\conv(\cS')\subseteq\conv(\cS)$.
Then, noting that $A'(\gamma^*) = A(\gamma^*) + \delta I_n\succ 0$, we may apply Theorem~\ref{thm:pd_implies_conv_hull_general} to the perturbed system to get $(\hat x,\hat t+\epsilon)\in\cR' = \conv(\cS')\subseteq\conv(\cS)$ as desired.

First note that $A_1 = A_1'$ so that $\Gamma$ is bounded if and only if $\Gamma'$ is bounded. We will then pick $\delta>0$ small enough such that
\begin{align*}
\delta\norm{\hat x}^2 \leq \frac{\epsilon}{2},\qquad \abs{\gamma_-' -\gamma_-}\abs{q_1(\hat x)}\leq \frac{\epsilon}{2},\qquad \abs{\gamma_+' -\gamma_+}\abs{q_1(\hat x)}\leq \frac{\epsilon}{2},
\end{align*}
where the last condition is only required when $\gamma_+$ and $\gamma_+'$ both exist. This is possible as the expression on the left of each inequality is continuous in $\delta$ and is strictly satisfied when $\delta =0$.

The following computation shows that $q'(\gamma_-', \hat x) \leq \hat t + \epsilon$.
\begin{align*}
q'(\gamma_-', \hat x) -(\hat t+ \epsilon)&= q'(\gamma_-, \hat x) -(\hat t+ \epsilon) + (\gamma_-'  - \gamma_-) q_1(\hat x)\\
&\leq q(\gamma_-, \hat x) + \delta \norm{\hat x}^2 -(\hat t+ \epsilon) + \abs{\gamma_-' - \gamma_-}\abs{q_1(\hat x)}\\
&\leq q(\gamma_-,\hat x) - \hat t\\
&\leq 0
\end{align*}
The first inequality follows by noting $q'(\gamma,x) = q(\gamma,x) + \delta\norm{x}^2$, the second inequality follows from our assumptions on $\delta$, and the third inequality follows from the assumption that $(\hat x, \hat t)\in \cS(\gamma_-)$. Thus $(\hat x,\hat t+\epsilon)\in\cS'(\gamma_-')$.
When $\Gamma$ is bounded (or equivalently, when $\gamma_+'$ and $\gamma_+$ exist), a similar calculation shows that $q'(\gamma_+',\hat x) - (\hat t + \epsilon)\leq 0$ so that $(\hat x,\hat t+\epsilon)\in\cS'(\gamma'_+)$.
Finally, when $\Gamma$ is unbounded we have $q_1'(\hat x) = q_1(\hat x)\leq 0$ so that $(\hat x,\hat t+\epsilon)\in\set{(x,t):\, q_1'(x)\leq 0}$. Thus, $(\hat x,\hat t+\epsilon)$ is in $\cR'$, concluding the proof.
\end{proof} \section{Estimation of the regularity parameters}
\label{app:estimating_regularity}

In Section \ref{sec:alg} we gave algorithms to solve the GTRS assuming that we had access to $(\xi,\zeta)$ and $\hat\gamma$ satisfying Assumption \ref{as:alg_regularity}.
In this appendix, we show how to compute these quantities.

Let $q_0,q_1$ satisfy Assumption \ref{as:alg}. Recall the definitions
\begin{align*}
\xi^* &\coloneqq \min\set{1,\max_{\gamma\geq 0} \lambda_{\min}(A(\gamma))},\hspace{2em}
\zeta^* \coloneqq \max\set{1,\gamma_+}.
\end{align*}
We will find $(\xi,\zeta)$ satisfying
\begin{align*}
\xi^*/4\leq \xi\leq \xi^*,\hspace{2em}
\zeta^* \leq \zeta \leq 4\zeta^*
\end{align*}
and a $\hat\gamma$ such that $\lambda_{\min}(A(\hat\gamma))\geq \xi$.

We will accomplish this in two stages. We begin by estimating $\xi^*$ using only an upper bound $\bar\zeta$ of $\zeta^*$. Then using our estimate $\xi$ we will compute $\zeta$.

\subsection{Computing $\xi$ and $\hat\gamma$}
\begin{algorithm}
	\caption{$\text{TestXi}(q_0,q_1,\xi,\bar\zeta, p_{\xi})$}
	\label{alg:test_xi}
	Given $q_0,q_1$ satisfying Assumption \ref{as:alg}, a guess $\xi$, an upper bound $\bar\zeta\geq\zeta^*$, and a failure probability $p_{\xi}>0$
	\begin{enumerate}
		\item Let $s_0 = 0$ and $t_0 = \bar\zeta$
		\item Let $T = \ceil{\log\kappa}+2$ where $\kappa={\bar\zeta / \xi}$
		\item For $k = 0,\dots,T-1$
		\begin{enumerate}
			\item Let $x = \text{ApproxEig}(A(s_k),2\bar\zeta, \xi/4, \frac{p_\xi}{3T})$. If $x^\top A(s_k)x \geq 3\xi/4$, then return $\hat\gamma = s_k$.
			\item  Let $x = \text{ApproxEig}(A(t_k),2\bar\zeta, \xi/4, \frac{p_\xi}{3T})$. If $x^\top A(t_k)x \geq 3\xi/4$, then return $\hat\gamma =t_k$.
			\item Let $\bar\gamma = (s_k+t_k)/2$
			\item Let $x = \text{ApproxEig}(A(\bar\gamma),2\bar\zeta,\xi/4, \frac{p_\xi}{3T})$. If $x^\top A(\bar\gamma)x\geq 3\xi/4$, then return $\hat\gamma =\bar\gamma$.
			\item If $x^\top A_1 x\geq 0$, let $s_{k+1} =\bar \gamma$ and $t_{k+1} = t_k$. Else, let $s_{k+1} = s_k$ and $t_{k+1} = \bar\gamma$.
		\end{enumerate}
		\item Return ``Fail''
	\end{enumerate}
\end{algorithm}
We start with the following guarantee for the algorithm TestXi (Algorithm~\ref{alg:test_xi}).
\begin{lemma}
\label{lem:test_xi}
Given $q_0,q_1$ satisfying Assumption \ref{as:alg}, an arbitrary $0<\xi\leq 1$, an upper bound $\bar\zeta\geq\zeta^*$, and a failure probability $p_{\xi}>0$, TestXi (Algorithm \ref{alg:test_xi}) will output
\begin{align*}
\begin{cases}
	\hat\gamma \text{ such that } \lambda_{\min}(A(\hat\gamma))\geq \xi/2 & \text{if } \xi\leq \xi^*\\
	\hat\gamma \text{ such that } \lambda_{\min}(A(\hat\gamma))\geq \xi/2 \text{ or ``Fail''}  & \text{if } \xi^*<\xi\leq 2\xi^*\\
	\text{``Fail''} & \text{if }2\xi^*<\xi
\end{cases}
\end{align*}
with probability $1-p_{\xi}$. This algorithm runs in time 
\begin{align*}
\tilde O\left(N\sqrt{\frac{\bar\zeta}{\xi}} \log\left(\frac{n}{p_\xi}\right)\log\left(\frac{\bar\zeta}{\xi}\right)\right).
\end{align*}
\end{lemma}
\begin{proof}
We condition on the event that ApproxEig succeeds every time it is called. By the union bound, this happens with probability at least $1-p_\xi$.

As we have conditioned on ApproxEig succeeding, any $\hat\gamma$ which is output by TestXi will satisfy
\begin{align*}
\lambda_{\min}(A(\hat\gamma))\geq 3\xi/4 - \xi/4 =\xi/2.
\end{align*}
It is clear that TestXi will output ``Fail'' if $\xi>2\xi^*$ as there does not exist any $\hat\gamma$ such that $\lambda_{\min}(A(\hat\gamma))\geq\xi^*$.
It remains to show that, given $\xi\leq \xi^*$, TestXi will output some $\hat\gamma$.

For the sake of contradiction, assume that the algorithm fails to output in each of the $T$ rounds.
Let $P \coloneqq \set{\gamma :\, \lambda_{\min}(A(\gamma))\geq 3\xi^*/4}$.
Recall that $\lambda_{\min}(A(\gamma))$ is $1$-Lipschitz in $\gamma$.
As there exists some $\gamma$ such that $\lambda_{\min}(A(\gamma))\geq \xi^*$ (see Definition~\ref{def:regularity}), we conclude that $P$ is an interval of length at least $\xi^*/2$.

Note that $P\subseteq[s_0,t_0]$. We will inductively show that $P\subseteq[s_k,t_k]$ for each $k\in\set{1,\dots,T}$. Let $k\in\set{0,\dots,T-1}$ and let $s_k,\bar\gamma,t_k$ be defined as in the algorithm and let $x$ be the unit vector found in step 3.(d).
We claim that $x^\top A_1x\neq 0$. Indeed suppose $x^\top A_1x = 0$, then $x^\top A(\gamma)x  = x^\top A(\bar\gamma)x \leq 3\xi/4$ for all $\gamma$. This contradicts the assumption that there exists some $\gamma$ such that $\lambda_{\min}(A(\gamma))\geq \xi$.
Now suppose $\gamma\in P$, then
\begin{align*}
\frac{3\xi^*}{4} &\leq x^\top A(\gamma)x
= x^\top A(\bar\gamma)x + (\gamma-\bar\gamma)x^\top A_1 x 
\leq \frac{3\xi^*}{4} +(\gamma-\bar\gamma)x^\top A_1x,
\end{align*}
where the first inequality follows from $\gamma\in P$, and the last one from the fact that the algorithm did not output in iteration $k$ (and thus the if statement in step 3.(d) did not hold).
Thus, if $x^\top A_1x > 0$, then we have the implication $\gamma\in P \implies \gamma\geq \bar\gamma$. Similarly, if $x^\top A_1x<0$, then we have the implication $\gamma\in P\implies \gamma\leq \bar\gamma$.
Then by induction, we have $P\subseteq [s_{k+1},t_{k+1}]$.

We conclude that $P$, an interval of length at least $\xi^*/2$, is contained in $[s_T,t_T]$ an interval of length
\begin{align*}
t_T - s_T = \frac{t_0 - s_0}{2^T}
\leq \xi/4.
\end{align*}
Noting that $\xi\leq\xi^*$ gives us the desired contradiction.

The running time of this algorithm follows from Lemma \ref{lem:approx_eig}.
\end{proof}

Given a lower bound $\xi\leq \xi^*$, Lemma \ref{lem:test_xi} guarantees that TestXi will find a $\hat\gamma$ satisfying $\lambda_{\min}(A(\hat\gamma))\geq \xi/2$ with high probability. In order to make use of this lemma \textit{without} a lower bound on $\xi^*$, we will simply repeatedly call TestXi with decreasing guesses for $\xi$. Consider Algorithm \ref{alg:approx_xi}.

\begin{algorithm}
	\caption{$\text{ApproxXi}(q_0,q_1,\bar\zeta,p)$}
	\label{alg:approx_xi}
	Given $q_0,q_1$ satisfying Assumption \ref{as:alg}, an upper bound $\bar\zeta\geq\zeta^*$, and failure probability $p>0$
	\begin{enumerate}
		\item For $k = 1,2,\dots$
		\begin{enumerate}
			\item Run $\text{TestXi}(q_0, q_1,  2^{-(k-1)},  \bar\zeta, 2^{-k}p)$.
			\item If TestXi outputs ``Fail'' then continue.
			\item Else, let $\hat\gamma$ be the output of TestXi and let $\xi = 2^{-k}$; return $\xi$ and $\hat\gamma$.
		\end{enumerate}
	\end{enumerate}
\end{algorithm}

\begin{theorem}
Given $q_0,q_1$ satisfying Assumption \ref{as:alg}, an upper bound $\bar\zeta\geq\zeta^*$, and a failure probability $p>0$, ApproxXi (Algorithm \ref{alg:approx_xi}) will output $\xi$ and $\hat\gamma$ such that
\begin{align*}
\xi^*/4\leq \xi\leq\xi^*,\hspace{2em} \lambda_{\min}(A(\hat\gamma))\geq \xi
\end{align*}
and run in time
\begin{align*}
\tilde O\left(N\sqrt{\frac{\bar\zeta}{\xi^*}}\log\left(\frac{n}{p}\right)\log\left(\bar\zeta\right)\log\left(\frac{1}{\xi^*}\right)^3\right)
\end{align*}
with probability $1-p$.
\end{theorem}
\begin{proof}
We condition on the event that TestXi succeeds every time it is called. By the union bound, this happens with probability at least $1-p$.

Let $k^*\in\set{1,2,\dots}$ be such that $\xi^*/2\leq 2^{-k^*}<\xi^*$.
Then, as we have conditioned on TestXi succeeding, Lemma \ref{lem:test_xi} guarantees that TestXi$(q_0, q_1, 2^{-k},\bar\zeta, 2^{-(k+1)}p)$ outputs
\begin{align*}
\begin{cases}
	\hat\gamma \text{ such that } \lambda_{\min}(A(\hat\gamma))\geq 2^{-k} & \text{if } 2^{-k}\leq \xi^*/2\\
	\hat\gamma \text{ such that } \lambda_{\min}(A(\hat\gamma))\geq 2^{-k} \text{ or ``Fail''}  & \text{if } \xi^*/2<2^{-k}\leq \xi^*\\
	\text{``Fail''} & \text{if }\xi^*<2^{-k}.
\end{cases}
\end{align*}
Thus, TestXi will output ``Fail'' for every $k<k^*$ and will output $\hat\gamma$ either on round $k^*$ or $k^*+1$. We can then bound
\begin{align*}
\lambda_{\min}(A(\hat\gamma)) \geq 2^{-(k^* +1)}\geq \frac{\xi^*}{4}.
\end{align*}

We bound the run time of the algorithm as follows.
\begin{align*}
&\sum_{k=1}^{k^* +1} \tilde O\left(N\sqrt{\frac{\bar\zeta}{2^{-(k-1)}}} \log\left(\frac{n}{2^{-k}p}\right)\log\left(\frac{\bar\zeta}{2^{-(k-1)}}\right)\right)\\
&= \tilde O\left(k^{*3} N\sqrt{\frac{\bar\zeta}{2^{-k^*}}} \log\left(\frac{n}{p}\right)\log\left(\bar\zeta\right) \right)\\
&= \tilde O\left(N\sqrt{\frac{\bar\zeta}{\xi^*}} \log\left(\frac{n}{p}\right)\log\left(\bar\zeta\right) \log\left(\frac{1}{\xi^*}\right)^3\right).\qedhere
\end{align*}
\end{proof}

\subsection{Computing $\zeta$}

Recall the guarantee of the algorithm ApproxGammaPlus.
\ApproxGammaPlusLemma*

We will repeatedly call ApproxGammaPlus with different choices of $\delta$. Consider the algorithm ApproxZeta.

\begin{algorithm}
	\caption{$\text{ApproxZeta}(q_0, q_1, \xi,\bar\zeta,\hat\gamma, p)$}
	\label{alg:approx_zeta}
	Given $q_0,q_1$ satisfying Assumption \ref{as:alg}, $(\xi,\bar\zeta)$ and $\hat\gamma$ satisfying Assumption \ref{as:alg_regularity}, and failure probability $p>0$
	\begin{enumerate}
		\item For $k = 1,2,\dots$
		\begin{enumerate}
			\item Let $\hat\zeta_k$ be the output of $\text{ApproxGammaPlus}(q_0,q_1,\xi, 2^{-(k-1)}\bar\zeta,\hat\gamma, 2^{-(k+1)}\bar\zeta ,2^{-k}p)$
			\item If $\hat\zeta_k \leq 2^{-(k+1)}\bar\zeta $ then continue
			\item Else set $\zeta \coloneqq  2^{-(k-1)}\bar\zeta$; return $\zeta$
		\end{enumerate}
	\end{enumerate}
\end{algorithm}

\begin{theorem}
Given $q_0,q_1$ satisfying Assumption \ref{as:alg}, $(\xi,\bar\zeta)$ and $\hat\gamma$ satisfying Assumption \ref{as:alg_regularity}, and failure probability $p>0$, ApproxZeta (Algorithm \ref{alg:approx_zeta}) will output $\zeta$ such that
\begin{align*}
\zeta^*\leq \zeta\leq 4\zeta^*
\end{align*}
and run in time
\begin{align*}
\tilde O\left(\frac{N\sqrt{\zeta^*}}{\sqrt{\xi}}\log\left(\frac{n}{p}\right)\log\left(\frac{1}{\xi}\right)\log\left(\frac{\bar\zeta}{\zeta^*}\right)^2\right)
\end{align*}
with probability $1-p$.
\end{theorem}
\begin{proof}
We condition on the event that ApproxGammaPlus succeeds every time it is called. By the union bound, this happens with probability at least $1-p$.

We first check that the assumptions of Lemma \ref{lem:tilde_gamma} hold. For $k = 1$, we have $2^{-(k-1)}\bar\zeta = \bar\zeta\geq \zeta^*$. Then by induction, and conditioning on ApproxGammaPlus succeeding, Lemma \ref{lem:tilde_gamma} guarantees
\begin{align*}
\zeta^* \leq \hat\zeta_k + 2^{-(k+1)}\bar\zeta.
\end{align*}
If ApproxZeta fails to terminate in round $k$, then 1.(b) ensures $\hat\zeta_k\leq 2^{-(k+1)}\bar\zeta$. This in turn implies that $\zeta^* \leq 2^{-((k+1)-1)}\bar\zeta$ and, by induction, the assumptions of Lemma~\ref{lem:tilde_gamma} hold in every round that ApproxGammaPlus is called.

Let $k$ be the round in which the algorithm terminates. If $k = 1$, then the guarantee of Lemma \ref{lem:tilde_gamma} implies $\zeta^*\geq\hat\zeta_1$, whence
\begin{align*}
\bar\zeta\geq\zeta^*\geq \hat\zeta_1 > \frac{1}{4}\bar\zeta.
\end{align*}
Thus, we may assume $k\geq 2$.
The condition of step 1.(b) then guarantees the two inequalities
\begin{align}
\label{eq:hat_zeta_guarantees}
\hat\zeta_{k-1} \leq 2^{-k}\bar\zeta,
\,\text{ and }\,
\hat\zeta_k > 2^{-(k+1)}\bar \zeta.
\end{align}
Then, we have
\begin{align*}
\zeta^* &\geq \hat\zeta_k
> 2^{-(k+1)}\bar\zeta 
= \frac{1}{4}\left(2^{-k}\bar\zeta + 2^{-k}\bar\zeta\right)
\geq \frac{1}{4} \left(\hat\zeta_{k-1} +2^{-((k-1)+1)}\bar\zeta\right)
\geq\zeta^*/4
\end{align*}
where the first and fifth relations follow from Lemma \ref{lem:tilde_gamma} and the second and fourth relations follow from \eqref{eq:hat_zeta_guarantees} above.

It remains to bound the run time of ApproxZeta.
Let $k^*\in\set{1,2,\dots}$ be such that $\zeta^* \leq 2^{-(k^*-1)}\bar\zeta < 2\zeta^*$. We show that ApproxZeta terminates within $k^*$ rounds. Suppose ApproxZeta reaches the $k^*$th round. Then, we have 
\begin{align*}
\hat\zeta_{k^*} &\geq \zeta^* - 2^{-(k^*+1)}\bar\zeta
> 2^{-k^*}\bar\zeta - 2^{-(k^*+1)}\bar\zeta
= 2^{-(k^* +1)}\bar\zeta,
\end{align*}
where we used Lemma \ref{lem:tilde_gamma} in the first relation, and the definition of $k^*$ in the second relation. Therefore, ApproxZeta terminates in round $k^*$  at the latest and we can bound the run time of this algorithm as
\begin{align*}
&\sum_{k=1}^{k^*} \tilde O\left(\frac{2^{-(k-1)}N\bar\zeta}{\sqrt{2^{-(k+1)}\xi  \bar\zeta}}\log\left(\frac{n}{2^{-k}p}\right)\log\left(\frac{2^{-(k-1)}\bar\zeta}{2^{-(k+1)}\xi\bar\zeta}\right)\right)\\
&=\tilde O\left(k^{*2}\frac{N\sqrt{2^{-k^*}\bar\zeta}}{\sqrt{\xi}}\log\left(\frac{n}{p}\right)\log\left(\frac{1}{\xi}\right)\right)\\
&=\tilde O\left(\frac{N\sqrt{\zeta^*}}{\sqrt{\xi}}\log\left(\frac{n}{p}\right)\log\left(\frac{1}{\xi}\right)\log\left(\frac{\bar\zeta}{\zeta^*}\right)^2\right).\qedhere
\end{align*}
\end{proof}

\end{document}